\documentclass[11pt]{article} 

\usepackage{geometry}
\usepackage{microtype}
\usepackage[dvipsnames]{xcolor}
\usepackage{color}

\usepackage{amssymb}
\usepackage{amsmath}
\usepackage{amsthm} 
\usepackage{amsfonts}
\usepackage{hyperref}
\usepackage[new]{old-arrows}
\usepackage[ruled]{algorithm2e}
\usepackage{tikz} \usetikzlibrary{arrows}
\usepackage{tikz-cd}
\definecolor{darkblue}{rgb}{0,0,0.4}

\newcommand {\mm}[1] {\ifmmode{#1}\else{\mbox{\(#1\)}}\fi}

\newcommand{\Xspace}        {\mm{{X}}}
\newcommand{\Sspace}        {\mm{{S}}}
\newcommand{\Yspace}        {\mm{{Y}}}

\newcommand{\Fcal}        {\mm{{\mathcal F}}}
\newcommand{\Lcal}        {\mm{{\mathcal L}}}
\newcommand{\Xstrata}  {\mm{{\mathfrak X}}} 
\newcommand{\bdr}  {\mm{{\partial}}} 
\newcommand{\str}  {\mbox{{St}}} 
\newcommand{\lk}  {\mbox{{Lk}}}
\newcommand{\closure}  {\mbox{Cl}}
\newcommand{\dime}  {\mbox{dim}}
\newcommand{\Top}{\mathbf{Top}}
\newcommand{\id}{\mm{\text{id}}}

\newtheorem{theorem}{Theorem}[section]
\theoremstyle{definition}
\newtheorem{definition}{Definition}[section]
\newtheorem{proposition}{Proposition}[section]

\newcommand{\para}[1]        {\vspace{2mm}\noindent{\textbf{#1}}}
\newcommand{\proofsketch}        {\noindent{\emph{Proof Sketch.~}}}

\title{Sheaf-Theoretic Stratification Learning \\from Geometric and Topological Perspectives}
\author{Adam Brown\thanks{E-mail: abrown@math.utah.edu.} \\ University of Utah 
\and Bei Wang\thanks{E-mail: beiwang@sci.utah.edu.} \\ University of Utah}
\date{}

\begin{document}

\begin{titlepage}

\maketitle 
\begin{abstract}
In this paper, we investigate a sheaf-theoretic interpretation of stratification learning from geometric and topological perspectives. 
Our main result is the construction of stratification learning algorithms framed in terms of a sheaf on a partially ordered set with the Alexandroff topology. We prove that the resulting decomposition is the unique minimal stratification for which the strata are homogeneous and the given sheaf is constructible. In particular, when we choose to work with the local homology sheaf, our algorithm gives an alternative to the local homology transfer algorithm given in Bendich et al. (2012), and the cohomology stratification algorithm given in Nanda (2017). Additionally, we give examples of stratifications based on the geometric techniques of Breiding et al. (2018), illustrating how the sheaf-theoretic approach can be used to study stratifications from both topological and geometric perspectives. This approach also points toward future applications of sheaf theory in the study of topological data analysis by illustrating the utility of the language of sheaf theory in generalizing existing algorithms.

\end{abstract}
\thispagestyle{empty}
\end{titlepage}

\pagestyle{plain}

\newpage
\setcounter{page}{1}
\section{Introduction}
\label{sec:introduction}
Our work is motivated by the following question: Given potentially high-dimensional point cloud samples, can we infer the structures of the underlying data? 
In the classic setting of \emph{manifold learning}, we often assume the support for the data is from a low-dimensional space with manifold structure. 
However, in practice, a significant amount of interesting data contains mixed dimensionality and singularities. 
To deal with this more general scenario, we assume the data are sampled from a mixture of possibly intersecting manifolds; the objective is to recover the different pieces, often treated as clusters, of the data associated with different manifolds of varying dimensions.
Such an objective gives rise to a problem of particular interest in the field of \emph{stratification learning}. 
Here, we use the word ``stratification learning'' loosely to mean an unsupervised, exploratory, clustering process that infers a decomposition of data into disjoint subsets that capture recognizable and meaningful structural information. 

Previous work in mathematics has focused on the study of stratified spaces under smooth and continuous settings~\cite{GoreskyMacPherson1988,Weinberger1994} without computational considerations of noisy and discrete datasets. 
Statistical approaches that rely on inferences of mixture models and local dimension estimation require strict geometric assumptions such as linearity~\cite{HaroRandallSapiro2005, LermanZhang2010, VidalMaSastry2005}, and may not handle general scenarios with complex singularities. 
Recently, approaches from topological data analysis~\cite{BendichCohen-SteinerEdelsbrunner2007, BendichWangMukherjee2012, SkrabaWang2014}, which rely heavily on ingredients from computational~\cite{EdelsbrunnerHarer2010} and intersection homology~\cite{Bendich2008, BendichHarer2011, GoreskyMacPherson1982}, are gaining momentum in stratification learning. 

Topological approaches transform the smooth and continuous setting favored by topologists to the noisy and discrete setting familiar to computational topologists in practice. 
In particular, the local structure of a point cloud (sampled from a stratified space) can be described by a multi-scale notion of local homology~\cite{BendichCohen-SteinerEdelsbrunner2007}; and the point cloud data could be clustered based on how the local homology of nearby sampled points map into one another~\cite{BendichWangMukherjee2012}.
Philosophically, our main goal is to find a stratification where any two points in the same strata (or cluster) can not be distinguished by homological methods, and any two points in different strata (different clusters) can be distinguished by homological methods. The majority of the paper will be spent developing a rigorous and computable interpretation of the purposely vague statement ``distinguished by homological methods''. Furthermore, we will see that our approach to computing the above stratification applies equally well to sheaves other than those based on local homology. As examples, we describe stratification learning with the combinatorially defined sheaf of maximal elements, and the geometrically motivated pre-sheaf of vanishing polynomials. 
This paper includes the full version of an extended abstract~\cite{BrownWang2018}, and further extends these results by exploring alternatives to homological stratifications which lie in the sheaf-theoretic framework (Sections \ref{sec:maximal-element} and \ref{sec:geometric-sheaf}).
As our work is an interplay between sheaf theory and stratification, we briefly review various notions of stratification before describing our results.  
\subsection{Stratifications}
Given a topological space $\Xspace$, a \emph{topological stratification} of $\Xspace$ is a finite filtration, that is, an increasing sequence of closed subspaces 
$ \emptyset = \Xspace_{-1} \subset \Xspace_0 \subset \cdots \subset \Xspace_d = \Xspace, $
such that for each $i$, $\Xspace_i - \Xspace_{i-1}$ is a (possibly empty) open $i$-dimensional topological manifold. See Figure~\ref{fig:pinchedtorus} for an example of a \emph{pinched torus}, that is, a torus with points along a geodesic with fixed longitude identified, and a spanning disc glued along the equator.

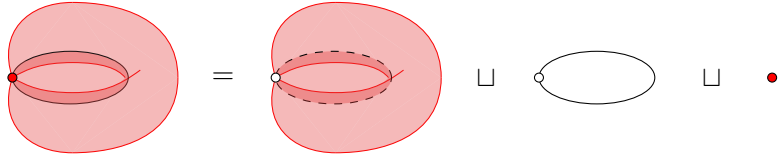
\begin{figure}[t!]

\begin{center}
\begin{tikzpicture}[circ/.style={
    circle,
    fill=white,
    draw,
    outer sep=0pt,
    inner sep=1.2pt
  },circ_red/.style={
    circle,
    fill=red,
    draw,
    outer sep=0pt,
    inner sep=1.2pt
  }]
\draw[red] (-.8,0) .. controls (-1,.6) and (-.6,1) .. (0,1);
\draw[red,xscale=-1] (-1.4,0) .. controls (-1.4,.8) and (-.6,1) .. (0,1);
\draw[red,rotate=180] (-1.4,0) .. controls (-1.4,.8) and (-.6,1) .. (0,1);
\draw[red,yscale=-1] (-.8,0) .. controls (-1,.6) and (-.6,1) .. (0,1);
\draw[red] (-.8,0) .. controls (-.6,-0.12) and (-.4,-0.2) .. (0,-.2) .. controls (.4,-0.2) and (.6,-0.12) .. (.9,0.1);
\draw[red] (-.8,0) .. controls (-.6,0.12) and (-.4,0.2) .. (0,.2) .. controls (.4,0.2) and (.6,0.12) .. (.7,0);
\fill[fill={rgb:brown,20;red,20;magenta,5}, opacity=.35] (0,1)--(-.8,0) .. controls (-.6,0.12) and (-.4,0.2) .. (0,.2) .. controls (.4,0.2) and (.6,0.12) .. (.7,0)--(0,1) -- cycle;
\fill[fill={rgb:brown,20;red,20;magenta,5}, opacity=.35] (-.03,0) ellipse (.77 and .35);
\draw[black, opacity=.9](-.03,0) ellipse (.77 and .35);
\fill[fill={rgb:brown,20;red,20;magenta,5}, opacity=.35] (0,-1)--(-.8,0) .. controls (-.6,-0.12) and (-.4,-0.2) .. (0,-.2) .. controls (.4,-0.2) and (.6,-0.12) .. (.7,0)--(0,-1) -- cycle;
\fill[fill={rgb:brown,20;red,20;magenta,5}, opacity=.35] (0,-1)--(.7,0)--(0,1)--(1.4,0)--(0,-1);

\fill[fill={rgb:brown,20;red,20;magenta,5}, opacity=.35] (-.8,0) .. controls (-1,.6) and (-.6,1) .. (0,1)--cycle;
\fill[fill={rgb:brown,20;red,20;magenta,5}, opacity=.35,xscale=-1] (-1.4,0) .. controls (-1.4,.8) and (-.6,1) .. (0,1)--cycle;
\fill[fill={rgb:brown,20;red,20;magenta,5}, opacity=.35,rotate=180] (-1.4,0) .. controls (-1.4,.8) and (-.6,1) .. (0,1)--cycle;
\fill[fill={rgb:brown,20;red,20;magenta,5}, opacity=.35,yscale=-1] (-.8,0) .. controls (-1,.6) and (-.6,1) .. (0,1)--cycle;

\draw (-.8,0) node[circ_red] {};

\draw (2,0) node {$=$};

\begin{scope}[shift={(3.5,0)}]
\draw[red] (-.8,0) .. controls (-1,.6) and (-.6,1) .. (0,1);
\draw[red,xscale=-1] (-1.4,0) .. controls (-1.4,.8) and (-.6,1) .. (0,1);
\draw[red,rotate=180] (-1.4,0) .. controls (-1.4,.8) and (-.6,1) .. (0,1);
\draw[red,yscale=-1] (-.8,0) .. controls (-1,.6) and (-.6,1) .. (0,1);
\draw[red] (-.8,0) .. controls (-.6,-0.12) and (-.4,-0.2) .. (0,-.2) .. controls (.4,-0.2) and (.6,-0.12) .. (.9,0.1);
\draw[red] (-.8,0) .. controls (-.6,0.12) and (-.4,0.2) .. (0,.2) .. controls (.4,0.2) and (.6,0.12) .. (.7,0);
\fill[fill={rgb:brown,20;red,20;magenta,5}, opacity=.35] (0,1)--(-.8,0) .. controls (-.6,0.12) and (-.4,0.2) .. (0,.2) .. controls (.4,0.2) and (.6,0.12) .. (.7,0)--(0,1) -- cycle;
\fill[fill={rgb:brown,20;red,20;magenta,5}, opacity=.35] (-.03,0) ellipse (.77 and .35);
\draw[dashed,black](-.03,0) ellipse (.77 and .35);
\fill[fill={rgb:brown,20;red,20;magenta,5}, opacity=.35] (0,-1)--(-.8,0) .. controls (-.6,-0.12) and (-.4,-0.2) .. (0,-.2) .. controls (.4,-0.2) and (.6,-0.12) .. (.7,0)--(0,-1) -- cycle;
\fill[fill={rgb:brown,20;red,20;magenta,5}, opacity=.35] (0,-1)--(.7,0)--(0,1)--(1.4,0)--(0,-1);

\fill[fill={rgb:brown,20;red,20;magenta,5}, opacity=.35] (-.8,0) .. controls (-1,.6) and (-.6,1) .. (0,1)--cycle;
\fill[fill={rgb:brown,20;red,20;magenta,5}, opacity=.35,xscale=-1] (-1.4,0) .. controls (-1.4,.8) and (-.6,1) .. (0,1)--cycle;
\fill[fill={rgb:brown,20;red,20;magenta,5}, opacity=.35,rotate=180] (-1.4,0) .. controls (-1.4,.8) and (-.6,1) .. (0,1)--cycle;
\fill[fill={rgb:brown,20;red,20;magenta,5}, opacity=.35,yscale=-1] (-.8,0) .. controls (-1,.6) and (-.6,1) .. (0,1)--cycle;

\draw (-.8,0) node[circ] {};
\end{scope}
\draw (5.5,0) node {$\sqcup$};
\begin{scope}[shift={(7,0)}]
\draw[black](-.03,0) ellipse (.77 and .35);
\draw (-.8,0) node[circ] {};
\end{scope}
\draw (8.5,0) node {$\sqcup$};
\draw (9.3,0) node[circ_red] {};
\end{tikzpicture}
\caption{Example of a topological stratification of a \emph{pinched torus}.}
\label{fig:pinchedtorus}
\end{center}
\end{figure}

\begin{figure}[t!]
\begin{center}
\begin{tikzpicture}
\fill[fill=red!85!white,nearly transparent](0,0) ellipse (1 and .5);
\fill[fill=red!85!white,nearly transparent] (0,0)--(1,0)--(1,1)--cycle;
\draw(0,0) ellipse (1 and .5);
\draw[black] (0,0)--(1,0)--(1,1)--cycle;
\node[circle,fill=black,inner sep=0pt,minimum size=3.5pt] (a) at (0,0) {};
\node[circle,fill=red,inner sep=0pt,minimum size=3pt] (a) at (0,0) {};
\node[circle,fill=black,inner sep=0pt,minimum size=3.5pt] (a) at (1,0) {};
\node[circle,fill=red,inner sep=0pt,minimum size=3pt] (a) at (1,0) {};
\draw (1.5,0) node {$=$};
\begin{scope}[shift={(3,0)}]
\fill[fill=red!85!white,nearly transparent](0,0) ellipse (1 and .5);
\fill[fill=red!85!white,nearly transparent] (0,0)--(1,0)--(1,1)--cycle;
\draw[dashed] (0,0) ellipse (1 and .5);
\draw[dashed] (0,0)--(1,0)--(1,1)--cycle;
\node[circle,fill=black,inner sep=0pt,minimum size=3.5pt] (a) at (0,0) {};
\node[circle,fill=white,inner sep=0pt,minimum size=3pt] (a) at (0,0) {};
\node[circle,fill=black,inner sep=0pt,minimum size=3.5pt] (a) at (1,0) {};
\node[circle,fill=white,inner sep=0pt,minimum size=3pt] (a) at (1,0) {};
\draw (1.5,0) node {$\sqcup$};
\end{scope}
\begin{scope}[shift={(6,0)}]
\draw(0,0) ellipse (1 and .5);
\draw[black] (0,0)--(1,0)--(1,1)--cycle;
\node[circle,fill=black,inner sep=0pt,minimum size=3.5pt] (a) at (0,0) {};
\node[circle,fill=white,inner sep=0pt,minimum size=3pt] (a) at (0,0) {};
\node[circle,fill=black,inner sep=0pt,minimum size=3.5pt] (a) at (1,0) {};
\node[circle,fill=white,inner sep=0pt,minimum size=3pt] (a) at (1,0) {};
\draw (1.5,0) node {$\sqcup$};
\end{scope}
\begin{scope}[shift={(8,0)}]
\node[circle,fill=black,inner sep=0pt,minimum size=3.5pt] (a) at (0,0) {};
\node[circle,fill=red,inner sep=0pt,minimum size=3pt] (a) at (0,0) {};
\node[circle,fill=black,inner sep=0pt,minimum size=3.5pt] (a) at (1,0) {};
\node[circle,fill=red,inner sep=0pt,minimum size=3pt] (a) at (1,0) {};

\end{scope}
\end{tikzpicture}
\caption{Example of a homological stratification of a \emph{sundial}.}
\label{fig:sundial}
\end{center}
\end{figure}
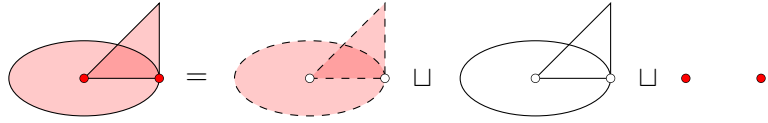
Ideally, we would compute a topological stratification for a given space. However, if we are restricted to using only homological methods, this is a dubious task. Topological invariants like homology are too rough to detect when a space such as $\Xspace_i-\Xspace_{i-1}$ is an open $i$-manifold. A well known example of a homological manifold which is not a topological manifold can be constructed from a homology 3-sphere with nontrivial fundamental group. The suspension of such a space is a homological manifold, but not a topological manifold, since the links of the suspension points have nontrivial fundamental groups~\cite[page 326]{Mio2000}. Specifically, the suspension of the Poincare's homology sphere is a homological manifold but not a topological manifold. In this paper, we will avoid the difficult problem of computing topological stratifications, and instead aim to investigate stratifications which can be computed using homological methods. Therefore, we must first consider a definition of stratification which does not rely on topological conditions which are not distinguished by homology. 
We begin with an extremely loose definition of stratification (Definition \ref{def:stratification}) which only requires the properties necessary to discuss the constructibility of sheaves (defined in Section \ref{subsec:sheaves}). We will then refine our definition of stratification by placing requirements on the constructibility of certain sheaves (Definition \ref{fstrat}).

\begin{definition}{\label{def:stratification}}

Given a topological space $\Xspace$, a \emph{stratification} $\mathfrak{X}$ of $\Xspace$ is a finite filtration of $\Xspace$ by closed subsets $\Xspace_i$:
$$ \emptyset = \Xspace_{-1} \subset \Xspace_0 \subset \cdots \subset \Xspace_d = \Xspace. $$
We refer to the space $X_i-X_{i-1}$ as \emph{stratum}, denoted by $S_i$, and a connected component of $S_i$ as a \emph{stratum piece}. 
\end{definition}

Suppose we have two stratifications of the topological space $\Xspace$, denoted $\Xstrata$  and $\Xstrata'$. 
We say that $\Xstrata$ is equivalent to $\Xstrata'$ if each stratum piece of $\Xstrata$ is equal to a stratum piece of $\Xstrata'$. 

\begin{definition}\label{def:coarser}
Given two inequivalent stratifications of $\Xspace$, $\Xstrata$ and $\Xstrata'$, we say $\Xstrata$ is \emph{coarser} than $\Xstrata'$, or $\Xstrata'$ \emph{refines} $\Xstrata$, if each stratum piece of $\Xstrata'$ is contained in a stratum piece of $\Xstrata$. 
\end{definition}

Figure~\ref{fig:coarser} illustrates some examples of stratifications which are coarser than those in Figure~\ref{fig:pinchedtorus} and Figure~\ref{fig:sundial}, as well as a different partition based on local homology transfer algorithm in~\cite{BendichWangMukherjee2012} for the sundial (bottom).

\para{Homological stratification.} 
There have been several approaches in topology literature to define homological stratifications. While proving the topological invariance of intersection homology, Goresky and MacPherson defined a type of homological stratification which they call a $\bar{p}$-stratification \cite[Section 4]{GoreskyMacPherson1983}. 
There have been several approaches for building on the ideas of Goresky and MacPherson, with applications to computational geometry and topology in mind (\cite{BendichHarer2011}, \cite{Nanda2017}). In this paper, we choose to adopt the perspective of homological stratifications found in~\cite{RourkeSanderson1999}, with a view toward sheaf theoretic generalizations and applications in topological data analysis. 

Consider a filtration $\emptyset= X_{-1}\subset X_0\subset\cdots\subset X_n=X$ of a topological space $X$. We can use relative homology to define a sheaf (the local homology sheaf) which associates to each an open set $U$ in $X_i$, the relative homology groups $H_\bullet(X_i,X_i-U)$, for some $i$.  Let $\mathcal{L}_i$ denote the local homology sheaf on the space $X_i$ (see Section \ref{subsec:LHS} for the definition of the local homology sheaf). 
\begin{definition}
We say that a stratification is a \emph{homological stratification} if $\mathcal{L}_n$ is locally constant (see Section \ref{subsec:sheaves} for the definition of locally constant)  when restricted to $X_i-X_{i-1}$, for each $i$. A stratification is a \emph{strong homological stratification} if for each $i$ both $\mathcal{L}_i$ and $\mathcal{L}_n$ are locally constant when restricted to $X_i-X_{i-1}$. Finally, a stratification is a \emph{very strong homological stratification} if for each $i$ and every $k\ge i$, $\mathcal{L}_k$ is locally constant when restricted to $X_i-X_{i-1}$.
\end{definition}  
As mentioned in \cite{RourkeSanderson1999}, it would be interesting to study the relationship between these definitions of homological stratification. We plan to pursue this in future work. The cohomological stratification given in~\cite{Nanda2017} can be considered as a cohomological anologue of the very strong homological stratification defined above. The utility of this definition is the extent to which it lends itself to the study of topological properties of individual strata. For example, it can be easily shown that the strata of such a stratification are $R$-(co)homology manifolds ($R$ being the ring with which the local cohomology is computed). The trade off for using the very strong homological stratification is in the number of local (co)homology groups which need to be computed. This is by far the most computationally expensive aspect of the algorithm, and the very strong homological stratification requires one to compute new homology groups for each sheaf $\mathcal{L}_i$. By contrast, the homological stratification only requires the computation of local homology groups corresponding to the sheaf $\mathcal{L}_n$. In this paper, we choose to study the less rigid (and more computable) notion of homological stratification (see Section \ref{sec:clustering} for more details on computing homological stratifications).

\begin{figure}[t!]
\begin{center}
\begin{center}
\begin{tikzpicture}[circ/.style={
    circle,
    fill=white,
    draw,
    outer sep=0pt,
    inner sep=1.2pt
  },circ_red/.style={
    circle,
    fill=red,
    draw,
    outer sep=0pt,
    inner sep=1.2pt
  }]
\draw[red] (-.8,0) .. controls (-1,.6) and (-.6,1) .. (0,1);
\draw[red,xscale=-1] (-1.4,0) .. controls (-1.4,.8) and (-.6,1) .. (0,1);
\draw[red,rotate=180] (-1.4,0) .. controls (-1.4,.8) and (-.6,1) .. (0,1);
\draw[red,yscale=-1] (-.8,0) .. controls (-1,.6) and (-.6,1) .. (0,1);
\draw[red] (-.8,0) .. controls (-.6,-0.12) and (-.4,-0.2) .. (0,-.2) .. controls (.4,-0.2) and (.6,-0.12) .. (.9,0.1);
\draw[red] (-.8,0) .. controls (-.6,0.12) and (-.4,0.2) .. (0,.2) .. controls (.4,0.2) and (.6,0.12) .. (.7,0);
\fill[fill={rgb:brown,20;red,20;magenta,5}, opacity=.35] (0,1)--(-.8,0) .. controls (-.6,0.12) and (-.4,0.2) .. (0,.2) .. controls (.4,0.2) and (.6,0.12) .. (.7,0)--(0,1) -- cycle;
\fill[fill={rgb:brown,20;red,20;magenta,5}, opacity=.35] (-.03,0) ellipse (.77 and .35);
\draw[black, opacity=.9](-.03,0) ellipse (.77 and .35);
\fill[fill={rgb:brown,20;red,20;magenta,5}, opacity=.35] (0,-1)--(-.8,0) .. controls (-.6,-0.12) and (-.4,-0.2) .. (0,-.2) .. controls (.4,-0.2) and (.6,-0.12) .. (.7,0)--(0,-1) -- cycle;
\fill[fill={rgb:brown,20;red,20;magenta,5}, opacity=.35] (0,-1)--(.7,0)--(0,1)--(1.4,0)--(0,-1);

\fill[fill={rgb:brown,20;red,20;magenta,5}, opacity=.35] (-.8,0) .. controls (-1,.6) and (-.6,1) .. (0,1)--cycle;
\fill[fill={rgb:brown,20;red,20;magenta,5}, opacity=.35,xscale=-1] (-1.4,0) .. controls (-1.4,.8) and (-.6,1) .. (0,1)--cycle;
\fill[fill={rgb:brown,20;red,20;magenta,5}, opacity=.35,rotate=180] (-1.4,0) .. controls (-1.4,.8) and (-.6,1) .. (0,1)--cycle;
\fill[fill={rgb:brown,20;red,20;magenta,5}, opacity=.35,yscale=-1] (-.8,0) .. controls (-1,.6) and (-.6,1) .. (0,1)--cycle;

\draw (-.8,0) node[circ_red] {};

\draw (2,0) node {$=$};

\begin{scope}[shift={(3.5,0)}]
\draw[red] (-.8,0) .. controls (-1,.6) and (-.6,1) .. (0,1);
\draw[red,xscale=-1] (-1.4,0) .. controls (-1.4,.8) and (-.6,1) .. (0,1);
\draw[red,rotate=180] (-1.4,0) .. controls (-1.4,.8) and (-.6,1) .. (0,1);
\draw[red,yscale=-1] (-.8,0) .. controls (-1,.6) and (-.6,1) .. (0,1);
\draw[red] (-.8,0) .. controls (-.6,-0.12) and (-.4,-0.2) .. (0,-.2) .. controls (.4,-0.2) and (.6,-0.12) .. (.9,0.1);
\draw[red] (-.8,0) .. controls (-.6,0.12) and (-.4,0.2) .. (0,.2) .. controls (.4,0.2) and (.6,0.12) .. (.7,0);
\fill[fill={rgb:brown,20;red,20;magenta,5}, opacity=.35] (0,1)--(-.8,0) .. controls (-.6,0.12) and (-.4,0.2) .. (0,.2) .. controls (.4,0.2) and (.6,0.12) .. (.7,0)--(0,1) -- cycle;
\fill[fill={rgb:brown,20;red,20;magenta,5}, opacity=.35] (-.03,0) ellipse (.77 and .35);
\draw[dashed,black](-.03,0) ellipse (.77 and .35);
\fill[fill={rgb:brown,20;red,20;magenta,5}, opacity=.35] (0,-1)--(-.8,0) .. controls (-.6,-0.12) and (-.4,-0.2) .. (0,-.2) .. controls (.4,-0.2) and (.6,-0.12) .. (.7,0)--(0,-1) -- cycle;
\fill[fill={rgb:brown,20;red,20;magenta,5}, opacity=.35] (0,-1)--(.7,0)--(0,1)--(1.4,0)--(0,-1);

\fill[fill={rgb:brown,20;red,20;magenta,5}, opacity=.35] (-.8,0) .. controls (-1,.6) and (-.6,1) .. (0,1)--cycle;
\fill[fill={rgb:brown,20;red,20;magenta,5}, opacity=.35,xscale=-1] (-1.4,0) .. controls (-1.4,.8) and (-.6,1) .. (0,1)--cycle;
\fill[fill={rgb:brown,20;red,20;magenta,5}, opacity=.35,rotate=180] (-1.4,0) .. controls (-1.4,.8) and (-.6,1) .. (0,1)--cycle;
\fill[fill={rgb:brown,20;red,20;magenta,5}, opacity=.35,yscale=-1] (-.8,0) .. controls (-1,.6) and (-.6,1) .. (0,1)--cycle;

\draw (-.8,0) node[circ] {};
\end{scope}
\draw (5.5,0) node {$\sqcup$};
\begin{scope}[shift={(7,0)}]
\draw[black](-.03,0) ellipse (.77 and .35);

\end{scope}
\begin{scope}[shift={(0,-2.2)}]
\fill[fill=red!85!white,nearly transparent](0,0) ellipse (1 and .5);
\fill[fill=red!85!white,nearly transparent] (0,0)--(1,0)--(1,1)--cycle;
\draw(0,0) ellipse (1 and .5);
\draw[black] (0,0)--(1,0)--(1,1)--cycle;
\node[circle,fill=black,inner sep=0pt,minimum size=3.5pt] (a) at (0,0) {};
\node[circle,fill=red,inner sep=0pt,minimum size=3pt] (a) at (0,0) {};
\node[circle,fill=black,inner sep=0pt,minimum size=3.5pt] (a) at (1,0) {};
\node[circle,fill=red,inner sep=0pt,minimum size=3pt] (a) at (1,0) {};
\draw (1.5,0) node {$=$};
\begin{scope}[shift={(3,0)}]
\fill[fill=red!85!white,nearly transparent](0,0) ellipse (1 and .5);
\fill[fill=red!85!white,nearly transparent] (0,0)--(1,0)--(1,1)--cycle;
\draw[dashed] (0,0) ellipse (1 and .5);
\draw[dashed] (0,0)--(1,0)--(1,1)--cycle;
\node[circle,fill=black,inner sep=0pt,minimum size=3.5pt] (a) at (0,0) {};
\node[circle,fill=white,inner sep=0pt,minimum size=3pt] (a) at (0,0) {};
\node[circle,fill=black,inner sep=0pt,minimum size=3.5pt] (a) at (1,0) {};
\node[circle,fill=white,inner sep=0pt,minimum size=3pt] (a) at (1,0) {};
\draw (1.5,0) node {$\sqcup$};
\end{scope}
\begin{scope}[shift={(6,0)}]
\draw(0,0) ellipse (1 and .5);
\draw[black] (0,0)--(1,0)--(1,1)--cycle;
\node[circle,fill=black,inner sep=0pt,minimum size=3.5pt] (a) at (0,0) {};
\node[circle,fill=red,inner sep=0pt,minimum size=3pt] (a) at (0,0) {};
\node[circle,fill=black,inner sep=0pt,minimum size=3.5pt] (a) at (1,0) {};
\node[circle,fill=white,inner sep=0pt,minimum size=3pt] (a) at (1,0) {};
\draw (1.5,0) node {$\sqcup$};
\end{scope}
\begin{scope}[shift={(9,0)}]

\node[circle,fill=black,inner sep=0pt,minimum size=3.5pt] (a) at (1,0) {};
\node[circle,fill=red,inner sep=0pt,minimum size=3pt] (a) at (1,0) {};

\end{scope}
\end{scope}
\begin{scope}[shift={(0,-4)}]
\fill[fill=red!85!white,nearly transparent](0,0) ellipse (1 and .5);
\fill[fill=red!85!white,nearly transparent] (0,0)--(1,0)--(1,1)--cycle;
\draw(0,0) ellipse (1 and .5);
\draw[black] (0,0)--(1,0)--(1,1)--cycle;
\node[circle,fill=black,inner sep=0pt,minimum size=3.5pt] (a) at (0,0) {};
\node[circle,fill=red,inner sep=0pt,minimum size=3pt] (a) at (0,0) {};
\node[circle,fill=black,inner sep=0pt,minimum size=3.5pt] (a) at (1,0) {};
\node[circle,fill=red,inner sep=0pt,minimum size=3pt] (a) at (1,0) {};
\draw (1.5,0) node {$=$};
\begin{scope}[shift={(3,0)}]
\fill[fill=red!85!white,nearly transparent](0,0) ellipse (1 and .5);
\draw[dashed] (0,0)--(1,0);
\draw[dashed] (0,0) ellipse (1 and .5);
\node[circle,fill=black,inner sep=0pt,minimum size=3.5pt] (a) at (0,0) {};
\node[circle,fill=red,inner sep=0pt,minimum size=3pt] (a) at (0,0) {};
\draw (1.5,0) node {$\sqcup$};

\begin{scope}[shift={(3,0)}]
\draw[dashed] (0,0)--(1,0)--(1,1)--cycle;
\fill[fill=red!85!white,nearly transparent] (0,0)--(1,0)--(1,1)--cycle;
\node[circle,fill=black,inner sep=0pt,minimum size=3.5pt] (a) at (1,0) {};
\node[circle,fill=white,inner sep=0pt,minimum size=3pt] (a) at (1,0) {};
\node[circle,fill=black,inner sep=0pt,minimum size=3.5pt] (a) at (0,0) {};
\node[circle,fill=white,inner sep=0pt,minimum size=3pt] (a) at (0,0) {};
\draw (1.5,0) node {$\sqcup$};
\end{scope}
\end{scope}
\begin{scope}[shift={(9,0)}]
\draw(0,0) ellipse (1 and .5);
\draw[black] (0,0)--(1,0)--(1,1)--cycle;
\node[circle,fill=black,inner sep=0pt,minimum size=3.5pt] (a) at (0,0) {};
\node[circle,fill=white,inner sep=0pt,minimum size=3pt] (a) at (0,0) {};
\node[circle,fill=black,inner sep=0pt,minimum size=3.5pt] (a) at (1,0) {};
\node[circle,fill=red,inner sep=0pt,minimum size=3pt] (a) at (1,0) {};

\end{scope}

\end{scope}
\end{tikzpicture}
\end{center}
\vspace{4mm}
\caption{Top: A topological stratification of a \emph{pinched torus} which is coarser than the stratification shown in Figure~\ref{fig:pinchedtorus}. Readers familiar with stratification theory will notice that while this is a topological stratification, it is not a locally conelike topological stratified space, also called a cs-space or cs-stratification~\cite[Definition 1.1]{HabeggerSaper1991}.
Middle: An example of a topological stratification of a \emph{sundial} which is coarser than the stratification shown in Figure~\ref{fig:sundial}. One can check that the coarser stratification is no longer a homological stratification.
Bottom: A partition of a \emph{sundial} into stratum pieces, based on local homology transfer in~\cite{BendichWangMukherjee2012}, where we assume an arbitrarily dense sampling of points from the sundial. Notice that the resulting decomposition is not a stratification by our definition, since it is not induced by a filtration by closed subspaces.}
\label{fig:coarser}
\end{center}
\end{figure}

\para{Sheaf theoretic stratification.}
The definition of homological stratification naturally lends itself to generalizations, which we now introduce (while delaying formal definition of constructible sheaves to Section~\ref{subsec:sheaves}). The primary observation which illuminates this generalization is that the fundamental mathematical structures (the association of a group or set to an open set and the restriction maps induced by inclusions of open sets) used to construct homological stratifications in \cite{GoreskyMacPherson1983,RourkeSanderson1999,Nanda2017} are exactly the formal properties which define a presheaf. Therefore, we will explore analogous constructive stratifications where local homology is replaced with an arbitrary sheaf or presheaf. The key notion used in the following definition is that of a \emph{constructible} sheaf. Intuitively, constructibility means that the sheaf can be decomposed into sheaves which are locally constant on various pieces, or stratum, of the given topological space.  
\begin{definition}
\label{fstrat}
Suppose $\Fcal$ is a sheaf on a topological space $\Xspace$. An \emph{$\Fcal$-stratification} (``sheaf-stratification'') of $\Xspace$ is a stratification such that $\Fcal$ is constructible with respect to $\Xspace=\amalg S_i$. A \emph{coarsest $\Fcal$-stratification} is an $\Fcal$-stratification such that $\Fcal$ is not constructible with respect to any coarser stratification.
\end{definition}
For general topological spaces, a coarsest $\Fcal$-stratification may not exist, and may not be unique if it does exist. The main focus of this paper will be proving existence and uniqueness results for certain coarsest $\Fcal$-stratifications. The extra structure needed to prove uniqueness is homogeneity of strata, and minimality of stratifications. By working with homogeneous stratifications, we are requiring strata to have a good notion of dimension. By defining a preorder on the set of homogeneous stratification (Section \ref{sec:results}), we are able to identify the output of the stratification algorithm described in Section \ref{sec:clustering} as the unique minimal homogeneous stratification with respect to the preorder.

\subsection{Our Contribution} 
In this paper, we study stratification learning using the tool of constructible sheaves. 
As a sheaf is designed to systematically track locally defined data attached to the open sets of a topological space, it seems to be a natural tool in the study of stratification based on local structure of the data. 
Our contributions are four-fold:
\begin{enumerate}
\item We prove the existence of coarsest $\Fcal$-stratifications and the existence and uniqueness of the minimal homogeneous $\Fcal$-stratification for finite $T_0$-spaces (Section~\ref{sec:results}). 
\item We give an algorithm for computing each of the above stratifications of a finite $T_0$-spaces based on a sheaf-theoretic language (Section~\ref{sec:clustering}).
\item In particular, when applying the local homology sheaf in our algorithm, we obtain a coarsest homological stratification (Section \ref{subsec:localhomology}). 
\item We give detailed examples of sheaf-theoretic stratifications based on combinatorial techniques (Section \ref{sec:maximal-element}) and geometric techniques (Section \ref{sec:geometric-sheaf}). 
\end{enumerate}
We envision that our abstraction could give rise to a larger class of stratification beyond homological stratification.
For instance, we give examples of a ``maximal element-stratification'' when the sheaf is defined by considering maximal elements of an open set (Section \ref{sec:maximal-element}) and a ``vanishing polynomial'' stratification when the (pre)-sheaf is defined with sets of vanishing polynomials (Section \ref{sec:geometric-sheaf}). Moreover, we see the geometric stratifications based on vanishing sets of polynomials as having natural applications to the mapper algorithm of~\cite{CarriereOudot17,MunchWang16,SinghMemoliCarlsson2007} (see Section \ref{sec:geometric-sheaf}).

\para{Comparison to prior work.}
This paper can be viewed as a continuation of previous works which adapt the stratification and homology theory of Goresky and MacPherson to the realm of topological data analysis. In \cite{RourkeSanderson1999}, Rourke and Sanderson give a proof of the topological invariance of intersection homology on PL homology stratifications, and give an recursive process for identifying a homological stratification (defined in Section 5 of \cite{RourkeSanderson1999}). In \cite{BendichHarer2011}, Bendich and Harer introduce a persistent version of intersection homology that can be applied to simplicial complexes. In \cite{BendichWangMukherjee2012}, Bendich, Wang, and Mukherjee provide computational approach that yields a stratification of point clouds by computing transfer maps between local homology groups of an open covering of the point cloud. In \cite{Nanda2017}, Nanda uses the machinery of derived categories to study cohomological stratifications based on local cohomology. 

Motivated by the results of \cite{Nanda2017} and \cite{BendichWangMukherjee2012}, we aim to develop a computational approach to the stratifications studied in \cite{RourkeSanderson1999}. Our main results can be summarized as the generalization of homological stratifications of \cite{RourkeSanderson1999} to $\Fcal$-stratifications, and a proof of existence and uniqueness of the minimal homogeneous $\Fcal$-stratification of a finite simplicial complex. When $\Fcal$ is the local homology sheaf, we recover the homological stratification described by \cite{RourkeSanderson1999}. While admitting a similar flavor as \cite{Nanda2017}, our work differs from \cite{Nanda2017} in several important ways. The most obvious difference is our choice to work with homology and sheaves rather than cohomology and cosheaves. More importantly however, we reduce the number of local homology groups which need to be computed.  
We will investigate the differences between homological, strong homological, and very strong homological stratifications in future work. 

In \cite{BendichWangMukherjee2012}, stratifications of point clouds are defined using persistent local homology and an equivalence based on local homology transfer. This method can easily be adapted to the setting of simplicial complexes by defining two simplicies $\tau$ and $\sigma$ to be in the same stratum if there exists a sequence of face/coface relations $\tau \le \gamma_1\ge \cdots \le \gamma_n \ge \sigma$, such that each face/coface relation induces an isomorphism of local homology groups. We choose not to take this approach here because the local homology transfer algorithm fails to produce reasonable stratifications for certain topological spaces. For example, the stratification of the \emph{sundial} example (i.e.~a stratified space with boundary) given by local homology transfer (see Figure \ref{fig:coarser}) is not technically a stratification by our definition, since the 1-dimensional stratum (which in this case is equal to $X_1$ in the induced filtration) is not closed. In comparison, the current algorithm correctly gives a coarsest homological stratification of this space. The algorithms described in this paper differ from local homology transfer in that we inductively define strata by requiring all restriction maps (transfer maps in a small neighborhood of a point/simplex) of a point/simplex in said stratum to induce isomorphisms of local homology groups. In this sense, our algorithm assigns a point/simplex to the top stratum if the local homology is unchanged as we move small amounts in ``any direction'' of the point, while the homology transfer algorithm assigns two points to the same stratum if there exists at least one path connecting them which induces an isomorphism of local homology.
Furthermore, the work in~\cite{BendichWangMukherjee2012} uses persistent homology in an essential way so that it is
amenable to point cloud data. The current work only brushes with the concept of  persistence in Section~\ref{sec:discussion}.
We plan to build on the results of \cite{BendichWangMukherjee2012} and \cite{SkrabaWang2014}, and extend the sheaf-theoretic stratification learning perspective described in this paper to the study of stratifications of point cloud data using persistent local homology.

\section{Preliminaries}
\label{sec:background}

\subsection{Compact Polyhedra, Finite $T_0$-spaces and Posets}
\label{subsec:t0}

Our broader aim is to compute a clustering of a finite set of points sampled from a compact polyhedron, based on the coarsest $\Fcal$-stratification of a finite $T_0$-space built from the point set. In this paper, we avoid discussion of sampling theory, and assume the finite point set forms the vertex set of a triangulated compact polyhedron. The finite $T_0$-space is the set of simplices of the triangulation, with the corresponding partial order given by the face/coface relation ($\tau\le\sigma$ if simplex $\tau$ is a face of simplex $\sigma$). To describe this correspondence in more detail, we first consider the connection between compact polyhedra and finite simplicial complexes. We then consider the correspondence between simplicial complexes and $T_0$-topological spaces.

\para{Compact polyhedra and triangulations.}
 A \emph{compact polyhedron} is a topological space which is homeomorphic to a finite simplicial complex. 
A \emph{triangulation} of a compact polyhedron is a finite simplicial complex $K$ and a homeomorphism from $K$ to the polyhedron.  

\para{$T_0$-spaces.}
 A \emph{$T_0$-space} is a topological space such that for each pair of distinct points, there exists an open set containing one but not the other. 
The correspondence between finite $T_0$-spaces and simplicial complexes is detailed in~\cite{McCord1978}: 
\begin{enumerate}
\item For each finite $T_0$-space $\Xspace$ there exists a (finite) simplicial complex $K$ and a weak homotopy equivalence $f: |K| \rightarrow \Xspace$. 
\item For each finite simplicial complex $K$ there exists a finite $T_0$-space $\Xspace$ and a weak homotopy equivalence $f: |K| \rightarrow \Xspace$. 
\end{enumerate}
Here, weak homotopy equivalence is a continuous map which induces isomorphisms on all homotopy groups. 

\para{$T_0$-spaces have a natural partial order.}
In this paper, we study certain topological properties of a compact polyhedron by considering its corresponding finite $T_0$-space. 
The last ingredient, developed in~\cite{Alexandroff1937}, is a natural partial order defined on a given finite $T_0$-space. We can define this partial ordering on a finite $T_0$-space $X$ by considering minimal open neighborhoods of each point (i.e. element) $x\in X$. Let $\Xspace$ be a finite $T_0$-space. 
Each point $x \in \Xspace$ has a minimal open neighborhood, denoted $B_x$, which is equal to the intersection of all open sets containing $x$.
$$
B_x=\bigcap_{U\in \mathcal{N}_x}U, 
$$
where $\mathcal{N}_x$ denotes the set of open sets containing $x$. Since $X$ is a finite space, there are only finitely many open sets. In particular, $\mathcal{N}_x$ is a finite set. So $B_x$ is defined to be the intersection of finitely many open sets, which implies that $B_x$ is an open neighborhood of $x$. Moreover, any other open neighborhood $V$ of $x$ must contain $B_x$ as a subset. We can define the partial ordering on $\Xspace$ by setting $x \le y$ if $B_y\subseteq B_x$. 

Conversely, we can endow any poset $\Xspace$ with the Alexandroff topology as follows.
For each element $\tau \in \Xspace$, we define a minimal open neighborhood containing $\tau$ by $B_\tau:=\{\gamma\in \Xspace: \gamma\ge \tau\}$. The collection of minimal open neighborhoods for each $\tau \in \Xspace$ forms a basis for a topology on $\Xspace$. We call this topology the Alexandroff topology. Moreover, a finite $T_0$-space $X$ is naturally equal (as topological spaces) to $X$ viewed as a poset with the Alexandroff topology. Therefore, we see that each partially ordered set is naturally a $T_0$-space, and each finite $T_0$-space is naturally a partially ordered set. The purpose for reviewing this correspondence here is to give the abstractly defined finite $T_0$-spaces a concrete and familiar realization.

As a concrete example, let $K$ to denote the finite $T_0$-space consisting of elements which are open simplices in a simplicial complex. 
In this setting, we can describe the $T_0$-space $K$ using the more familiar language of simplicial complexes. 
This is also the setting that applies to most of the subsequent examples in this paper. 
For a simplex $\sigma \in K$, its minimal open neighborhood $B_{\sigma}$ is its \emph{star} consisting of all cofaces of $\sigma$, 
$\str(\sigma) = \{\tau \in K \mid \sigma \leq \tau\}$. 
 Using the partial order induced by inclusions of minimal open neighborhoods, we set $\tau\le \sigma$ if $B_\sigma\subseteq B_\tau$. By describing minimal open neighborhoods as open stars of simplices, we can state the partial order as $\tau\le \sigma$ if $\str(\sigma)\subseteq \str(\tau)$. 
 On the other hand, $K$ is equipped with a partial order based on face relations, where $\tau \leq \sigma$ if simplex $\tau$ is a face of simplex $\sigma$. 
 Therefore, the two partial orders (the face partial order and the open neighborhood inclusion partial order) coincide.

Given a finite $T_0$-space $X$ with the above partial order, we say $x_0\le x_1\le \cdots\le x_n$ (where $x_i\in X$) is a maximal chain in $X$ if there is no totally ordered subset $Y\subset X$ consisting of elements $y_j\in Y$ such that $y_0\le \cdots\le y_j\le \cdots\le y_k$ and $\cup_{i=0}^n\{x_i\}\subsetneq Y$. The cardinality of a chain $x_0\le x_1\le \cdots\le x_n$ is $n+1$. We say that a finite $T_0$-space has dimension $m$ if the maximal cardinality of maximal chains is $m+1$. 
\begin{definition}
An $m$-dimensional simplicial complex is called \emph{homogeneous} if each simplex of dimension less than $m$ is a face of a simplex of dimension $m$. Motivated by the correspondence between simplicial complexes and $T_0$-spaces, we say an $m$-dimensional finite $T_0$-space is \emph{homogeneous} if each maximal chain has cardinality $m+1$.
\end{definition}

\begin{figure}[ht!]
\begin{center}
\begin{tikzpicture}[scale=.65,circ/.style={
    circle,
    fill=white,
    draw,
    outer sep=0pt,
    inner sep=1.5pt
  }]

\draw (0,0) node {$\bullet$};
\draw (1,1.75) node{$\bullet$};
\draw (2,0) node{$\bullet$};
\draw (1,-1.75) node{$\bullet$};
\draw (-1,-1.75) node{$\bullet$};
\draw (0,0)--(2,0);
\draw (0,0)--(1,1.75);
\draw (1,1.75)--(2,0);
\draw (0,0)--(1,-1.75);
\draw (1,-1.75)--(-1,-1.75);
\draw (-1,-1.75)--(0,0);
\draw[draw=black,fill=red, nearly transparent] (0,0)--(1,-1.75)--(-1,-1.75);
\draw[draw=black,fill=red, nearly transparent] (0,0)--(1,1.75)--(2,0);

\draw (2.75,0) node{$:$};

\begin{scope}[shift={(3.5,0)}]

\draw[dashed] (0,0)--(2,0);
\draw[dashed] (0,0)--(1,1.75);
\draw[dashed] (1,1.75)--(2,0);
\draw[dashed] (0,0)--(1,-1.75);
\draw[dashed] (1,-1.75)--(-1,-1.75);
\draw[dashed] (-1,-1.75)--(0,0);
\fill[fill=red, nearly transparent] (0,0)--(1,1.75)--(2,0);
\draw[draw=black,fill=red, nearly transparent] (0,0)--(1,-1.75)--(-1,-1.75);
\node[circle,fill=black,inner sep=0pt,minimum size=4pt]  at (0,0) {};
\node[circle,fill=white,inner sep=0pt,minimum size=3pt]  at (0,0) {};
\node[circle,fill=black,inner sep=0pt,minimum size=4pt]  at (2,0) {};
\node[circle,fill=white,inner sep=0pt,minimum size=3pt]  at (2,0) {};
\node[circle,fill=black,inner sep=0pt,minimum size=4pt]  at (1,1.75) {};
\node[circle,fill=white,inner sep=0pt,minimum size=3pt]  at (1,1.75) {};
\node[circle,fill=black,inner sep=0pt,minimum size=4pt]  at (-1,-1.75) {};
\node[circle,fill=white,inner sep=0pt,minimum size=3pt]  at (-1,-1.75) {};
\node[circle,fill=black,inner sep=0pt,minimum size=4pt]  at (1,-1.75) {};
\node[circle,fill=white,inner sep=0pt,minimum size=3pt]  at (1,-1.75) {};

\draw (2.75,0) node{$\sqcup$};
\end{scope}

\begin{scope}[shift={(7,0)}]
\draw (0,0) node {$\bullet$};
\draw (1,1.75) node{$\bullet$};
\draw (2,0) node{$\bullet$};
\draw (1,-1.75) node{$\bullet$};
\draw (-1,-1.75) node{$\bullet$};
\draw (0,0)--(2,0);
\draw (0,0)--(1,1.75);
\draw (1,1.75)--(2,0);
\draw (0,0)--(1,-1.75);
\draw (1,-1.75)--(-1,-1.75);
\draw (-1,-1.75)--(0,0);

\end{scope}
\draw (9.75,0) node{$,$};

\begin{scope}[shift={(10.5,0)}]

\draw[dashed] (0,0)--(2,0);
\draw[dashed] (0,0)--(1,1.75);
\draw[dashed] (1,1.75)--(2,0);
\fill[fill=red, nearly transparent] (0,0)--(1,1.75)--(2,0);
\node[circle,fill=black,inner sep=0pt,minimum size=4pt]  at (0,0) {};
\node[circle,fill=white,inner sep=0pt,minimum size=3pt]  at (0,0) {};
\node[circle,fill=black,inner sep=0pt,minimum size=4pt]  at (2,0) {};
\node[circle,fill=white,inner sep=0pt,minimum size=3pt]  at (2,0) {};
\node[circle,fill=black,inner sep=0pt,minimum size=4pt]  at (1,1.75) {};
\node[circle,fill=white,inner sep=0pt,minimum size=3pt]  at (1,1.75) {};

\draw (2.75,0) node{$\sqcup$};
\end{scope}

\begin{scope}[shift={(14,0)}]
\draw (0,0)--(2,0);
\draw (0,0)--(1,1.75);
\draw (1,1.75)--(2,0);
\draw (0,0)--(1,-1.75);
\draw (1,-1.75)--(-1,-1.75);
\draw (-1,-1.75)--(0,0);
\draw[draw=black,fill=red, nearly transparent] (0,0)--(1,-1.75)--(-1,-1.75);
\node[circle,fill=black,inner sep=0pt,minimum size=4pt]  at (0,0) {};
\node[circle,fill=black,inner sep=0pt,minimum size=4pt]  at (2,0) {};
\node[circle,fill=black,inner sep=0pt,minimum size=4pt]  at (1,1.75) {};
\node[circle,fill=black,inner sep=0pt,minimum size=4pt]  at (-1,-1.75) {};
\node[circle,fill=black,inner sep=0pt,minimum size=4pt]  at (1,-1.75) {};
\end{scope}
\end{tikzpicture}
\caption{An example of a homogeneous simplicial complex (left), a homogeneous stratification (middle), and a stratification which is not homogeneous (right).} 
\label{fig:homogeneous-example}
\end{center}
\end{figure}

See Figure~\ref{fig:homogeneous-example} for an example. 
The correspondences allow us to study certain topological properties of compact polyhedra by using the combinatorial theory of partially ordered sets. In particular, instead of using the more complicated theory of sheaves on the geometric realization $\vert K\vert$ of a simplicial complex $K$, we will continue by studying sheaves on the corresponding finite $T_0$-space, denoted by $\Xspace$. 


\subsection{Constructible Sheaves}
\label{subsec:sheaves}

Intuitively, a sheaf assigns some piece of data to each open set in a topological space $\Xspace$, in a way that allows us to glue together data to recover some information about the larger space. This process can be described as the mathematics behind understanding global structure by studying local properties of a space. In this paper, we are primarily interested in sheaves on finite $T_0$-spaces, which are closely related to the cellular sheaves studied in \cite{Shepard1985}, \cite{Curry2014}, and \cite{Nanda2017}. 

\para{Sheaves.}
Suppose $\Xspace$ is a topological space. Let $\Top(X)$ denote the category consisting of objects which are open sets in $X$ with morphisms given by inclusion. Let $\Fcal$ be a contravariant functor from $\Top(X)$ to $\mathcal{S}$, the category of sets. For open sets $U\subset V$ in $X$, we refer to the morphism $\Fcal(U\subset V): \Fcal(V)\rightarrow \Fcal(U)$ induced by $\Fcal$ and the inclusion $U\subset V$, as a \emph{restriction} map from $V$ to $U$. We say that $\Fcal$ is a \emph{sheaf} \footnote{See \cite[Chapter 2]{Curry2014}, \cite[Chapter 3]{Kirwan2006} for various introductions to sheaf theory. In~\cite{Curry2014}, Curry gives a very general definition of pre-sheaves which take values in a general category (possibly differing from the category of sets). Curry then discusses how certain properties of the chosen category (such as the existence of direct and inverse limits) imply properties for the corresponding pre-sheaves and sheaves. \cite{Kirwan2006} gives a standard introduction to sheaves, constructible sheaves, and intersection homology.} on $\Xspace$ if $\Fcal$ satisfies the following conditions 1.-3.; a \emph{presheaf} is a functor $\mathcal{E}$ (as above) which satisfies conditions 1.-2.: 
\begin{enumerate}
\item $\Fcal(U\subset U)=\id_U.$
\item If $U\subset V\subset W$, then $\Fcal(U\subset W)=\Fcal(U\subset V)\circ\Fcal(V\subset W)$.
\item If $\{V_i\}$ is an open cover of $U$, and $s_i\in\Fcal(V_i)$ has the property that $\forall i, j$, 
$\Fcal((V_i\cap V_j)\subset V_i)(s_i)=\Fcal((V_j\cap V_i)\subset V_j)(s_j)$, 
then there exists a unique $s \in \Fcal(U)$ such that $\forall i$, $\Fcal(V_i\subset U)(s)=s_i$. 
\end{enumerate}
There is a useful process known as \emph{sheafification}, which allows us to transform any presheaf into a sheaf. In the setting of finite $T_0$-spaces, sheafification takes on a relatively simple form. Let $\mathcal{E}$ be a presheaf on a finite $T_0$-space $X$. Then the sheafification of $\mathcal{E}$, denoted $\mathcal{E}^+$, is given by 
\begin{eqnarray*}
\mathcal{E}^+(U)=\bigg\{f:U\rightarrow \coprod_{x\in U} \mathcal{E}(B_x)\bigg\vert f(x)\in \mathcal{E}(B_x)&\text{ and }&f(y)=\Fcal(B_y\subset B_x)(f(x))\\
&&\text{ for all }y\ge x\bigg\}
\end{eqnarray*}
For any presheaf $\mathcal{E}$, it can be seen that $\mathcal{E}^+$ is necessarily a sheaf. We only need to know the values $\mathcal{E}(B_x)$ for minimal open neighborhoods $B_x$, and the corresponding restriction maps between minimal open neighborhoods $\mathcal{E}(B_x\subset B_y)$, in order to define the sheafification of $\mathcal{E}$. The result is that two presheaves will sheafify to the same sheaf if they agree on all minimal open neighborhoods. We will use this fact several times in Section \ref{sec:results}. Unless otherwise specified, for the remaining of this paper, we use $\Xspace$ to denote a $T_0$-space. 

\para{Pull back of a sheaf.}
For notational convenience, define for each subset $\Yspace \subset \Xspace$ the \emph{star} of $\Yspace$ by 
$\str(\Yspace) :=\cup_{y\in \Yspace}B_y$, 
where $B_y$ is the minimal open neighborhood of $y \in \Xspace$. We can think of the star of $Y$ as the smallest open set containing $Y$. 
Let $\Xspace$ and $\Yspace$ be two finite $T_0$-spaces. 
The following property can be thought of as a way to transfer a sheaf on 
$\Yspace$ to a sheaf on $X$ through a continuous map $f: \Xspace \rightarrow \Yspace$.
Let $\Fcal$ be a sheaf on $\Yspace$. Then the \emph{pull back} of $\Fcal$, denoted $f^{-1}\Fcal$, is defined to be the sheafification of the presheaf $\mathcal{E}$ which maps an open set $U\subset \Xspace$ to $\mathcal{E}(U) := \Fcal(\str(f(U)))$. We can avoid using direct limits in our definition of pull back because each point in a finite $T_0$-space has a minimal open neighborhood~\cite[Chapter 5]{Curry2014}. The pull back of $\Fcal$ along an inclusion map $\iota:U\hookrightarrow X$ is called the \emph{restriction} of $\Fcal$ to $U$, and is denoted $\Fcal\vert_U$.

\para{Constant and locally constant sheaves.}
Now we can define classes of well-behaved sheaves, constant and locally constant ones, which we can think of intuitively as analogues of constant functions based on definitions common to algebraic geometry and topology~\cite{Kirwan2006}. A sheaf $\Fcal$ is a \emph{constant sheaf} if $\Fcal$ is isomorphic to the pull back of a sheaf $\mathcal{G}$ on a single point space $\{x\}$, along the projection map $p:X\rightarrow x$. A sheaf $\Fcal$ is \emph{locally constant} if for all $x\in X$, there is a neighborhood $U$ of $x$ such that $\Fcal\vert_U$ (the restriction of $\Fcal$ to $U$), is a constant sheaf.

\begin{definition}
\label{const}
A sheaf $\Fcal$ on a finite $T_0$-space $X$ is \emph{constructible} with respect to the decomposition $X=\coprod S_i$ of $X$ into finitely many disjoint locally closed subsets, if $\Fcal\vert_{S_i}$ is locally constant for each $i$. 
\end{definition}

\section{Main Results}
\label{sec:results}

In this section we state three of our main theorems, namely, the existence of $\Fcal$-stratifications (Definition \ref{fstrat}, Proposition \ref{theorem:exists}), the existence of coarsest $\Fcal$-stratifications (Theorem \ref{theorem:coarsest}), and the existence and uniqueness of minimal homogeneous $\Fcal$-stratifications (Theorem \ref{theorem:uniqueness}). Of course, Theorem \ref{theorem:coarsest} immediately implies Proposition \ref{theorem:exists}. We choose to include a separate statement of Proposition \ref{theorem:exists} however, as we wish to illustrate the existence of $\Fcal$-stratifications which are not necessarily the coarsest. We include proof sketches here and refer to Section~\ref{sec:proofs} for technical details. 
 \begin{proposition}
 {\label{theorem:exists}}
 Let $\Fcal$ be a sheaf on a finite $T_0$-space $X$. There exists an $\Fcal$-stratification of $X$ (see Definition \ref{fstrat} and Definition \ref{const}).
 \end{proposition}
 
\proofsketch
 $\Fcal$ is constructible with respect to the decomposition $X=\coprod_{x\in X}x$. $\qed$

 
\begin{theorem}
  {\label{theorem:coarsest}}
 Let $\Fcal$ be a sheaf on a finite $T_0$-space $X$. There exists a coarsest $\Fcal$-stratification of $X$.
 \end{theorem}
 
\proofsketch
We can prove Theorem \ref{theorem:coarsest} easily as follows. There are only finitely many stratifications of our space $X$, which implies that there must be an $\Fcal$-stratification with a minimal number of strata pieces. Such a stratification must be a coarsest stratification, since any coarser stratification would have fewer strata pieces. 

However, the above proof is rather unenlightening if we are interested in computing the coarsest $\Fcal$-stratification. Therefore we include a constructive proof of the existence of a coarsest $\Fcal$-stratification which we sketch here. We can proceed iteratively, by defining the top-dimensional stratum to be the collection of points (i.e.~elements) so that the  sheaf is constant when restricted to the minimal open neighborhoods of the said points. 
Our process is a greedy algorithm as it seeks to make locally optimal choice at each stage.
Rephrased using the language of simplicial complexes, the top dimensional statum will consist of all simplices such that the sheaf is constant when restricted to the open star of the simplex. We then remove the top-dimensional stratum from our space, and pull back the sheaf to the remaining points/simplices. Notice that the points/simplices which are maximal with respect to the natural partial order (described in Section \ref{subsec:t0}) are guaranteed to be included in the top dimensional stratum. We proceed inductively until all the points in our space have been assigned to a stratum.
We can see that this is a coarsest $\Fcal$-stratification by arguing that this algorithm, in some sense, maximizes the size of each stratum piece, and thus any coarser $\Fcal$-stratification is actually equivalent to the one constructed above. We refer the reader to Section  \ref{ap:coarsest} for the details of the above argument. 
$\qed$

Notice that if $\Fcal$ is the trivial sheaf on $X$, then the coarsest $\Fcal$-stratification of $X$ will be the trivial stratification, consisting of a single stratum, even if the space $X$ has mixed dimensionality. For example, the coarsest $\Fcal$-stratification of the space $X$, as illustrated in Figure \ref{fig:coarsest_nonunique}, is the trivial stratification of $X$. To remedy this situation, we will introduce an extra constraint of homogeneity on a stratification, which will require each stratum piece to have ``pure'' dimension. 
\begin{definition}\label{def:homogeneous}
Suppose $\Fcal$ is a sheaf on a finite $T_0$-space $X$. A \emph{homogeneous $\Fcal$-stratification} is an $\Fcal$-stratification such that for each $i$, the closure of the stratum $S_i$ in $X_i$ is homogeneous of dimension $i$ (defined in Section \ref{subsec:t0}).\end{definition}
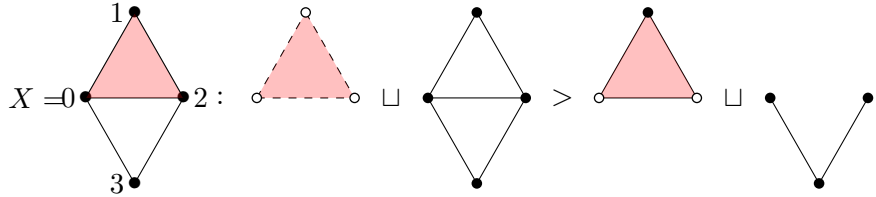
\begin{figure}[ht!]
\begin{center}
\begin{tikzpicture}[scale=.65,circ/.style={
    circle,
    fill=white,
    draw,
    outer sep=0pt,
    inner sep=1.5pt
  }]

\draw (0,0) node[left] {0};
\draw (-1,0) node {$X=$};
\draw (0,0) node {$\bullet$};
\draw (1,1.75) node[left] {1};
\draw (1,1.75) node{$\bullet$};
\draw (2,0) node[right] {2};
\draw (2,0) node{$\bullet$};
\draw (1,-1.75) node[left] {3};
\draw (1,-1.75) node{$\bullet$};
\draw (0,0)--(2,0);
\draw (0,0)--(1,1.75);
\draw (1,1.75)--(2,0);
\draw (0,0)--(1,-1.75);
\draw (1,-1.75)--(2,0);
\draw[draw=black,fill=red, nearly transparent] (0,0)--(1,1.75)--(2,0);

\draw (2.75,0) node{$:$};

\begin{scope}[shift={(3.5,0)}]

\draw[dashed] (0,0)--(2,0);
\draw[dashed] (0,0)--(1,1.75);
\draw[dashed] (1,1.75)--(2,0);
\fill[fill=red, nearly transparent] (0,0)--(1,1.75)--(2,0);
\node[circle,fill=black,inner sep=0pt,minimum size=4pt]  at (0,0) {};
\node[circle,fill=white,inner sep=0pt,minimum size=3pt]  at (0,0) {};
\node[circle,fill=black,inner sep=0pt,minimum size=4pt]  at (2,0) {};
\node[circle,fill=white,inner sep=0pt,minimum size=3pt]  at (2,0) {};
\node[circle,fill=black,inner sep=0pt,minimum size=4pt]  at (1,1.75) {};
\node[circle,fill=white,inner sep=0pt,minimum size=3pt]  at (1,1.75) {};

\draw (2.75,0) node{$\sqcup$};
\end{scope}

\begin{scope}[shift={(7,0)}]
\draw (0,0)--(2,0);
\draw (0,0)--(1,1.75);
\draw (1,1.75)--(2,0);
\draw (0,0)--(1,-1.75);
\draw (1,-1.75)--(2,0);
\node[circle,fill=black,inner sep=0pt,minimum size=4pt] at (0,0) {};
\node[circle,fill=black,inner sep=0pt,minimum size=3pt] at (0,0) {};
\node[circle,fill=black,inner sep=0pt,minimum size=4pt] at (2,0) {};
\node[circle,fill=black,inner sep=0pt,minimum size=3pt] at (2,0) {};
\node[circle,fill=black,inner sep=0pt,minimum size=4pt] at (1,1.75) {};
\node[circle,fill=black,inner sep=0pt,minimum size=4pt] at (1,-1.75) {};

\end{scope}
\draw (9.75,0) node{$>$};

\begin{scope}[shift={(10.5,0)}]

\draw (0,0)--(2,0);
\draw (0,0)--(1,1.75);
\draw (1,1.75)--(2,0);
\fill[fill=red, nearly transparent] (0,0)--(1,1.75)--(2,0);
\node[circle,fill=black,inner sep=0pt,minimum size=4pt]  at (0,0) {};
\node[circle,fill=white,inner sep=0pt,minimum size=3pt]  at (0,0) {};
\node[circle,fill=black,inner sep=0pt,minimum size=4pt]  at (2,0) {};
\node[circle,fill=white,inner sep=0pt,minimum size=3pt]  at (2,0) {};
\node[circle,fill=black,inner sep=0pt,minimum size=4pt] at (1,1.75) {};
\draw (2.75,0) node{$\sqcup$};
\end{scope}

\begin{scope}[shift={(14,0)}]
\draw (0,0)--(1,-1.75);
\draw (1,-1.75)--(2,0);
\node[circle,fill=black,inner sep=0pt,minimum size=4pt] at (0,0) {};
\node[circle,fill=black,inner sep=0pt,minimum size=4pt] at (2,0) {};
\node[circle,fill=black,inner sep=0pt,minimum size=4pt] at (1,-1.75) {};

\end{scope}

\end{tikzpicture}
\caption{An example of two inequivalent coarsest homogeneous $\mathcal{C}$-stratifications, where $\mathcal{C}$ is a constant sheaf on $X$. The stratification given on the right is the minimal homogeneous $\mathcal{C}$-stratification.} 
\label{fig:coarsest_nonunique}
\end{center}
\end{figure}
Observe that if $\Fcal$ is the constant sheaf on a simplicial complex $X$, as in Figure \ref{fig:coarsest_nonunique}, then the skeletal filtration of $X$ (defined by setting $X_i$ to be the set of $i$-simplices of the simplical complex $X$) is automatically a homogeneous $\Fcal$-stratification. However, the skeletal stratification is not guaranteed to be a coarsest $\Fcal$-stratification. 

 We will introduce a lexicographical preorder on the set of homogeneous stratifications of a finite $T_0$-space $\Xspace$. While this preorder can be extended to the set of all stratifications, we choose to define the preorder on homogeneous stratifications in order to take advantage of the required regularity the indexing of stratum. Specifically, the fact that the stratum $X_i-X_{i-1}$ is homogeneous of dimension $i$ will guarantee compatible indexing of each homogeneous stratification of $X$. Let $\Xstrata$ be a homogeneous stratification of $X$ given by
$$
\emptyset = X_{-1} \subset X_0\subset X_1\subset\cdots \subset X_n=X.
$$
We define a sequence $A^{\Xstrata}:=\{|X_n|,\cdots,|X_i|,\cdots, |X_0|\}$, where $|X_i|$ denotes the cardinality of the set $X_i$. Notice that the order of the sequence $A^\Xstrata$ is the reverse of the natural filtration order corresponding to $\Xstrata$. Given two homogeneous stratifications $\Xstrata$ and $\Xstrata'$ of $X$, let $m\in\{-1,\cdots,n\}$ be the smallest integer such that $|X'_j|-|X_j|=0$ for all $j>m$. We say that $\Xstrata<\Xstrata'$ if $|X'_m|-|X_m|$ is positive;  $\Xstrata\le\Xstrata'$ if $m=-1$. In other words, we say that $\Xstrata<\Xstrata'$ if the first nonzero term in the sequence $A^{\Xstrata'}-A^\Xstrata=\{|X'_n|-|X_n|,\cdots,|X'_i|-|X_i|,\cdots,|X'_0|-|X_0|\}$ is positive. For two stratifications $\Xstrata$ and $\Xstrata'$ of the same space $X$, the first term, $|X'_n|-|X_n|$, of the sequence $A^{\Xstrata'}-A^\Xstrata$, is automatically zero, since $X_n=X_n'=X$.

 Figure \ref{fig:coarsest_nonunique} illustrates two homogeneous stratifications of the simplicial complex $X$. The stratification on the left produces an integer sequence $A^{\Xstrata'}=\{10,9,0\}$, while the stratification on the right produces the sequence $A^{\Xstrata}=\{10,5,0\}$. The sequence of differences is then $A^{\Xstrata'}-A^\Xstrata=\{0,4,0\}$. The first nonzero term in the sequence is positive, so $\Xstrata<\Xstrata'$.

\begin{proposition}
If $\Xstrata$ and $\Xstrata'$ are homogeneous stratifications of $X$ such that $\Xstrata$ is coarser than $\Xstrata'$, then $\Xstrata\le \Xstrata'$. Moreover, if $\Xstrata$ is coarser than $\Xstrata'$, and $A^{\Xstrata}=A^{\Xstrata'}$, then $\Xstrata=\Xstrata'$. 
\end{proposition}

\begin{proof} 
By Definition \ref{def:coarser}, each stratum piece of $\Xstrata'$ is contained in a stratum piece of $\Xstrata$. Assume that a connected component of $X'_i-X'_{i-1}$ is contained in a connected component of $X_j-X_{j-1}$ for some $j$. Since both stratifications are assumed to be homogeneous, we have that each connected component of $X'_i-X'_{i-1}$ is homogeneous of dimension $i$ and each connected component of $X'_j-X'_{j-1}$ is homogeneous of dimension $j$. By the assumed containment of stratum pieces, we have that $i\le j$.

Assume that $X_j=X'_j$ for all $j\ge i$. Then $X_j-X_{j-1}=X'_j-X'_{j-1}$ for all $j>i$. Since $X_j-X_{j-1}=X_j'-X'_{j-1}$ for all $j>i $, and $(X'_j-X'_{j-1})\cap(X'_i-X'_{i-1})=\emptyset$ for all $j>i$, we must have that $(X_j-X_{j-1})\cap(X'_i-X'_{i-1})=\emptyset$ for all $j>i$. By the preceding argument, each connected component of $X'_i-X'_{i-1}$ is contained in $X_j-X_{j-1}$ for some $j\ge i$. Therefore, $X'_i-X'_{i-1}\subseteq X_i-X_{i-1}$. Combining the equality $X'_i=X_i$ with the containment $X'_i-X'_{i-1}\subseteq X_i-X_{i-1}$, we have that $X_{i-1}\subseteq X'_{i-1}$.

If $X_{i-1}\subsetneq X'_{i-1}$, then $|X_{i-1}'|-|X_{i-1}|>0$, and $\Xstrata <\Xstrata'$. Otherwise, $X_{i-1}=X'_{i-1}$, and the proof proceeds inductively until either $X_{k}\subsetneq X'_{k}$ for some $k$, or $X_k=X'_k$ for all $k$. If $X_{k}\subsetneq X'_{k}$ then $\Xstrata<\Xstrata'$. If $X_k=X'_k$ for all $k$, then $\Xstrata=\Xstrata'$. 
\end{proof}

We will use the preorder on the set of homogeneous $\Fcal$-stratifications to identify a unique stratification which is minimal in the preorder. 
\begin{definition}{\label{def:minimal}}
 We say that a stratification $\Xstrata$ is a \emph{minimal homogeneous $\Fcal$-stratification} if $\Xstrata\le \Xstrata'$ for every other homogeneous $\Fcal$-stratification $\Xstrata'$. In other words, a homogeneous $\Fcal$-stratification is \emph{minimal} if it is minimal with respect to the lexocographic order on homogeneous $\Fcal$-stratifications.
\end{definition}
There are several examples which illustrate the necessity of introducing \emph{minimal} homogeneous $\Fcal$-stratifications, rather than studying only coarsest homogeneous $\Fcal$-stratifications. Consider the simplicial complex $K$ illustrated in Figure \ref{fig:coarsest_nonunique}. Let $\mathcal{C}$ denote the constant sheaf on the corresponding $T_0$-space $X$ by assigning the one dimensional vector space $k$ to the star of each simplex $\sigma$, $\mathcal{C}(\str\sigma)=k$. For the restriction maps, $\mathcal{C}(\str(\tau)\subset\str(\sigma))$ is an isomorphism for each pair $\sigma<\tau$. Figure \ref{fig:coarsest_nonunique} shows that there are two coarsest homogeneous $\mathcal{C}$-stratifications of $X$. The stratification on the right side of Figure \ref{fig:coarsest_nonunique} is the minimal homogeneous $\mathcal{C}$-stratification.

 \begin{theorem}
 {\label{theorem:uniqueness}}
 Let $K$ be a finite simplicial complex, and $X$ be a finite $T_0$-space consisting of the simplices of $K$ endowed with the Alexandroff topology. Let $\Fcal$ be a sheaf on $X$. There exists a unique minimal homogeneous $\Fcal$-stratification of $X$.
\end{theorem}
\proofsketch
The idea for this proof is very similar to that of the Theorem \ref{theorem:coarsest}. We construct a stratification in a very similar way, with the only difference being that we must be careful to only construct homogeneous strata. The argument for the uniqueness of the resulting stratification uses the observation that this iterative process maximizes the size of the current stratum (starting with the top-dimensional stratum) before moving on to define lower-dimensional strata. Thus the resulting stratification is minimal in the lexocographic order. The top-dimensional stratum of any other minimal homogeneous $\Fcal$-stratification then must equal the top stratum constructed above, since these must both include the set of top-dimensional simplices, and have maximal size. An inductive argument then shows the stratifications are equivalent. Again, we refer readers to Section \ref{ap:uniqueness} for the remaining details. $\qed$

\section{A Sheaf-Theoretic Stratification Learning Algorithm}
\label{sec:clustering}

We outline an explicit algorithm for computing the coarsest $\Fcal$-stratification of a space $X$ given a particular sheaf $\Fcal$. 
We give two examples of stratification learning using the local homology sheaf (Section~\ref{sec:localhomology}) and the sheaf of maximal elements (Section~\ref{sec:maximal-element}).

Let $X$ be a finite $T_0$-space, equipped with the natural partial order defined in Section \ref{subsec:t0}.  
Instead of using the sheaf-theoretic language of Theorem \ref{theorem:uniqueness}, we frame the computation in terms of $X$ and an ``indicator function" $\delta$. 
For every $x,y\in X$ with a relation $x\le y$, $\delta$ assigns a binary value to the relation.
That is, $\delta(x\le y)=1$ if the restriction map $\Fcal(B_y\subset B_x):\Fcal(B_x)\rightarrow \Fcal(B_y)$ is an isomorphism, and $\delta(x\le y)=0$ otherwise. We say a pair $w\le y$ is \emph{adjacent} if $w\le z\le y$ implies $z=w$ or $z=y$ (in other words, there are no elements in between $w$ and $y$). Due to condition 2. in the definition of a sheaf (Section \ref{subsec:sheaves}), $\delta$ is fully determined by the values $\delta(w\le y)$ assigned to each adjacent pair $(w,y)$. If $a_1\le a_2\le \cdots\le a_k$ is a chain of adjacent elements ($a_i$ is adjacent to $a_{i+1}$ for each $i$), we have that $\delta(a_1 \le a_k)=\delta(a_1 \le a_2)\cdot \delta(a_2\le a_3)\cdots \delta(a_{k-1}\le a_k)$.
As $X$ is equipped with a finite partially ordering, computing $\delta$ can be interpreted as assigning a binary label to the edges of a Hasse diagram associated with the partial ordering (see Section~\ref{sec:localhomology} for an example).

For simplicity, we assume that $\delta$ is pre-computed, with a complexity of $O(m)$ where $m$ denotes the number of adjacent relations in $X$. When $X$ corresponds to a simplicial complex $K$, $m$ is the number of nonzero terms in the boundary matrices of $K$.  
$\delta$ can, of course, be processed on-the-fly, which may lead to more efficient algorithm. 
In addition, determining the value of $\delta$ is a local computation for each $x \in X$, therefore it is easily parallelizable. 

\para{Computing a coarsest $\Fcal$-stratification.}
 If we are only concerned with calculating a coarsest $\Fcal$-stratification as described in Theorem \ref{theorem:coarsest}, we may use the algorithm below.
 
  \begin{enumerate}
 \item Set $i=0$, $d_0=\text{dim}X$, $X_{d_0}=X$, and initialize $S_j = \emptyset$, for all $0\le j\le d_0$.
 \item While $d_i\ge 0$, do 
 \begin{enumerate}
 \item For each $x\in X_{d_i}$, set $S_{d_i}=S_{d_i}\cup x$ if $\delta(w\le y)=1$, $\forall$ adjacent pairs $w\le y$ in $ B_x\cap X_{d_i}$
 \item Set $d_{i+1}=\text{dim}(X_{d_i}-S_{d_i})$
 \item Define $X_{d_{i+1}}=X_{d_i}-S_{d_i}$
 \item Set $i=i+1$ 
 \end{enumerate}
 \item Return S
 \end{enumerate}
 
Here, $i$ is the step counter; $d_i$ is the dimension of the current  strata of interest; the set $S_{d_i}$ is the stratum of dimension $d_i$.
$d_i$ decreases from $\dime(X)$ to $0$.   
To include an element $x$ to the current stratum $S_{d_i}$, we need to check $\delta$ for adjacent relations among all $x$'s cofaces. It is worthwhile to note that in step 2a, we only want to consider adjacent pairs $w\le y$ in the minimal open neighborhood of $x$ (i.e. the open star of $x$) which have not been previously assigned to a stratum. In other words, instead of checking $\delta(w\le y)$ for every adjacent pair in the minimal open neighborhood of $x$, we only want to require $\delta(w\le y)=1$ for each adjacent pair which is contained in both the minimal open neighborhood of $x$ and contained in the $d_i$-level of the filtration. Moreover, each element of $X$ which is maximal with respect to the partial order will be automatically included in current strata defined by the algorithm. Hence the cardinality of the set of unassigned elements will decrease after each iteration of the algorithm.  
 
\para{Computing the unique minimal homogeneous $\Fcal$-stratification.}
If we would like to obtain the unique minimal homogeneous $\Fcal$-stratification, then we need to modify step 2a. Let $c(x,i)=1$ if all maximal chains in $X_{d_i}$ containing $x$ have cardinality $d_i$, and $c(x,i)=0$ otherwise. Then the modified version of 2.a. is: 
\begin{eqnarray*}
 \text{2.a'. For each}&&\text{ $x\in X_{d_i}$, set $S_{d_i}= S_{d_i}\cup x$ if} \\
 &&\delta(w\le y)=1 \forall \text{ adj. pairs }  w\le y \text{ in } B_x\cap X_{d_i}  \text{, and}\text{ $c(x,i)=1$}
\end{eqnarray*}
The only modification in the algorithm for computing the unique minimal homogeneous $\Fcal$-stratification is the additional requirement that $c(x,i)=1$ for $x$ to be included in the stratum $X_{d_i}$. This algorithm constructs a stratification which is minimal in the lexicographic order on homogeneous $\Fcal$-stratifications because the algorithm iteratively maximizes the cardinality of each stratum, starting with the top dimensional stratum. In Section \ref{ap:uniqueness}, we give a detailed proof that since the cardinality of the stratum is maximized at each iteration of the algorithm, the resulting stratification must be minimal in the lexicographic order.

\section{Stratification Learning with the Local Homology Sheaf}
\label{sec:localhomology}

\subsection{Local Homology Sheaf}
\label{subsec:LHS}

We begin with an example that uses homology to cluster simplices of a $T_0$-space into strata. In order to capture local structure, we study a version of relative homology (i.e.,~the local homology) involving minimal open neighborhoods. In order to define relative (and local) homology, we will start with a chain complex associated to a finite $T_0$-space. 


For a finite $T_0$-space $X$, consider the chain complex $C_\bullet(X)$, where $C_p(X)$ denotes the free $R$-module generated by $(p+1)$-chains in $X$, with chain maps $\bdr_p:C_p(X)\rightarrow C_{p-1}(X)$ given by 
$$
\bdr_p(a_0 \le \cdots \le a_p)=\sum (-1)^i (a_0\le \cdots \le \hat{a_i} \le \cdots \le a_p)
$$
where $\hat{a_i}$ means that the element $a_i$ is to be removed from the chain. Define the homology of $X$ (with coefficients in the ring $R$) to be the homology groups of $C_\bullet(X)$, $H_i(X)=\ker \bdr_i / \text{im } \bdr_{i+1}$. Since $X$ is a finite $T_0$-space, each subset $U$ of $X$ is a finite $T_0$-space. We therefore define the homology of $X$ relative to $U$, $H_\bullet(X,U)$, to be the homology groups of the quotient of chain complexes $C_\bullet(X)/C_\bullet(U)$.

If $X$ is a more general topological space (CW space, simplicial complex, manifold, etc), then the \emph{local homology} of $X$ at $x\in X$ is defined to be the direct limit of relative homology $H_\bullet(X, X-x):=\varinjlim H_\bullet (X,X-U)$ (where the direct limit is taken over all open neighborhoods $U$ of $x$ with the inclusion partial order)~\cite[page 196]{Munkres1984}. In our setting, the local homology of $X$ (a finite $T_0$-space) at a point $x\in X$ is given by $H_\bullet(X,X-B_x)$. Here we avoid using notions of direct limit by working with topological spaces that have minimal open neighborhoods. This motivates our decision to refer to the sheaf defined by relative homology $H_\bullet(X,X-U)$ for each open set $U$ (see Theorem \ref{theorem:LH-sheaf}), as the \emph{local homology sheaf} \footnote{See \cite{CianciOttina2017} for an interesting approach to the computation of homology groups of finite $T_0$-spaces using spectral sequences.}. 

The following theorem, though straightforward, provides justification for applying the results of Section~\ref{sec:clustering} to local homology computations.  
\begin{theorem}
\label{theorem:LH-sheaf}
The functor $\mathcal{L}$ from the category of open sets of a finite $T_0$-space to the category of graded $R$-modules, defined by
$$
 \mathcal{L}(U):=H_\bullet(X,X-U)
$$
 where $R$ is the ring of coefficients of the relative homology, is a sheaf on $X$. 
 \end{theorem}
 \begin{proof}
We first show that conditions $1.$-$2.$ are satisfied in the definition of sheaf from Section~\ref{sec:background}. The inclusion of open sets $U\subset V$, and equivalently $X-V\subset X-U$, induce a morphism of graded $R$-modules,  
$$
 \mathcal{L}(U\subset V):H_\bullet(X,X-V)\rightarrow H_\bullet(X,X-U).
$$
We have the following commutative diagram of chain complexes
\begin{center}
\begin{tikzcd}%
        0 \arrow[r] & C_\bullet(X-V) \arrow[r] \arrow[d] & C_\bullet(X) \arrow{r}[above]{p_1}  \arrow{d}[right]{\text{id}} & C_\bullet(X)/C_\bullet(X-V) \arrow{r} \arrow[d]  & 0 \\  
        0 \arrow[r] & C_\bullet(X-U) \arrow[r]  & C_\bullet(X) \arrow{r}[above]{p_2}  & C_\bullet(X)/C_\bullet(X-U) \arrow[r]  & 0 
    \end{tikzcd}
    \end{center}
    where the map $C_\bullet(X)/C_\bullet(X-V) \rightarrow  C_\bullet(X)/C_\bullet(X-U)$ is defined by $p_2\circ p_1^{-1}$, and is well-defined since $X-V\subset X-U$. For a triple $U\subset V\subset W$, we have the restriction maps
$$
H_\bullet(X, X-W)\rightarrow H_\bullet(X,X-V)\rightarrow H_\bullet(X,X-U)
$$
whose composition is equal to $H_\bullet(X, X-W)\rightarrow H_\bullet(X,X-U)$. This can be seen by applying our construction of the restriction map above to three short exact sequences of chain complexes.  In order to prove that condition $3.$ in the definition of a sheaf is satisfied, we could apply Mayer-Vietoris sequences for relative homology groups. But considering that we only need to think of $\mathcal{L}$ as a presheaf in order to apply our algorithm, we will not include the details of this part of the proof. 
\end{proof}

\subsection{An Example Using the Local Homology Sheaf}
\label{subsec:localhomology}
If $X$ is a $T_0$-space corresponding to a simplicial complex $K$, then the local homology groups in Section~\ref{subsec:LHS} are isomorphic to the simplicial homology groups of $K$. 
We now give a detailed example of stratification learning using local homology sheaf for the sundial example from Figure \ref{fig:sundial}. 
We will abuse notation slightly, and use $K$ to denote the finite $T_0$-space consisting of elements which are open simplices corresponding to the triangulated sundial (Figure~\ref{fig:triangulated-sundial}). We choose this notation so that we can describe our $T_0$-space using the more familiar language of simplicial complexes. 
For a simplex $\sigma \in K$, its minimal open neighborhood $B_{\sigma}$ is its \emph{star} consisting of all cofaces of $\sigma$, 
$\str(\sigma) = \{\tau \in K \mid \sigma \leq \tau\}$.   
The \emph{closed star}, $\overline{\str}(\sigma)$, is the smallest subcomplex that contains the star. 
The \emph{link} consists of all simplices in the closed star that are disjoint from the star, $\lk(\sigma) = \{\tau \in \overline{\str}(\sigma) \mid \tau \cap \str(\sigma) = \emptyset\}$. $K$ is equipped with a partial order  based on face relations, where $x < y$ if $x$ is a proper face of $y$. 
This partial order gives rise to a Hasse diagram illustrated in Figure~\ref{fig:sundial-hasse}.
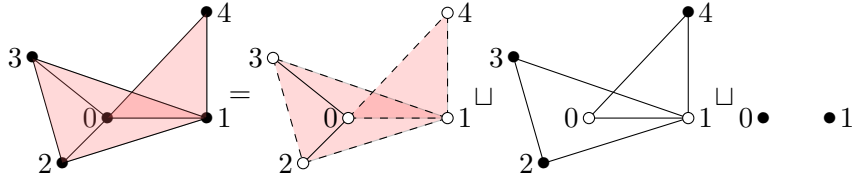
\begin{figure}[ht!]
\begin{center}
\begin{tikzpicture}[scale=.8,circ/.style={
    circle,
    fill=white,
    draw,
    outer sep=0pt,
    inner sep=1.5pt
  }]
\draw (0,0) node {$\bullet$};
\draw (0,0) node[left] {0};
\draw (1.65,0) node {$\bullet$};
\draw (1.65,0) node[right] {1};
\draw (1.65,1.75) node {$\bullet$};
\draw (1.65,1.75) node[right] {4};
\draw (-1.25,1) node {$\bullet$};
\draw (-1.25,1) node[left] {3};
\draw (-.75,-.75) node {$\bullet$};
\draw (-.75,-.75) node[left] {2};
\draw (0,0)--(1.65,0);
\draw (0,0)--(1.65,1.75);
\draw (1.65,0)--(1.65,1.75);
\draw (1.65,0)--(-1.25,1);
\draw (1.65,0)--(-.75,-.75);
\draw (-1.25,1)--(-.75,-.75);
\draw (0,0)--(-.75,-.75);
\draw (0,0)--(-1.25,1);
\draw[draw=black,fill=red!60!white, nearly transparent] (0,0)--(1.65,0)--(1.65,1.75);
\draw[draw=black,fill=red!60!white, nearly transparent] (-.75,-.75)--(1.65,0)--(-1.25,1);

\draw (2.2,.35) node{$=$};
\draw[draw=black,fill=red!60!white, nearly transparent] (4+0,0)--(4+1.65,0)--(4+1.65,1.75);
\draw[draw=black,fill=red!60!white, nearly transparent] (4+-.75,-.75)--(4+1.65,0)--(4+-1.25,1);
\draw (4+0,0) node[left] {0};
\draw (4+1.65,0) node[right] {1};
\draw (4+1.65,1.75) node[right] {4};
\draw (4+-1.25,1) node[left] {3};

\draw (4+-.75,-.75) node[left] {2};
\draw[dashed] (4+0,0)--(4+1.65,0);
\draw[dashed]  (4+0,0)--(4+1.65,1.75);
\draw[dashed] (4+1.65,0)--(4+1.65,1.75);
\draw[dashed] (4+1.65,0)--(4-1.25,1);
\draw[dashed]  (4+1.65,0)--(4+-.75,-.75);
\draw[dashed]  (4+-1.25,1)--(4+-.75,-.75);
\draw (4+0,0)--(4+-.75,-.75);
\draw (4+0,0)--(4+-1.25,1);

\draw (4+-.75,-.75) node[circ] {};
\draw (4+-1.25,1) node[circ] {};
\draw (4+1.65,1.75) node[circ] {};
\draw (4+1.65,0) node[circ] {};
\draw (4+0,0) node[circ]{};

\draw (6.25,.35) node{$\sqcup$};

\draw (8+0,0) node[left] {0};
\draw (8+1.65,0) node[right] {1};
\draw (8+1.65,1.75) node[right] {4};
\draw (8+-1.25,1) node[left] {3};
\draw (8+-.75,-.75) node[left] {2};
\draw (8+0,0)--(8+1.65,0);
\draw  (8+0,0)--(8+1.65,1.75);
\draw (8+1.65,0)--(8+1.65,1.75);
\draw (8+1.65,0)--(8+-1.25,1);
\draw  (8+1.65,0)--(8+-.75,-.75);
\draw (8+-1.25,1)--(8+-.75,-.75);

\draw (8+-.75,-.75) node[] {$\bullet$};
\draw (8+-1.25,1) node[] {$\bullet$};
\draw (8+1.65,1.75) node[] {$\bullet$};
\draw (8+1.65,0) node[circ] {};
\draw (8+0,0) node[circ]{};

\draw (10.25,.35) node{$\sqcup$};

\draw (10.9,0) node[left] {0};
\draw (11+1,0) node[right] {1};

\draw (11+1,0) node[] {$\bullet$};
\draw (10.9,0) node[]{$\bullet$};

\end{tikzpicture}
\caption{A triangulated sundial and its stratification based on the local homology sheaf.} 
\label{fig:triangulated-sundial}
\end{center}
\end{figure}

\begin{figure}[t!]
\begin{center}
\begin{tikzpicture}[scale=.75,circ/.style={
    circle,
    fill=red!30!white,
    draw,
    outer sep=0pt,
    inner sep=1.3pt
  }]
\draw (-5.5,0) node {$[0,1,3]$};

 \shade[ball color = red!40, opacity = 0.4] (-4.5,0) circle (.2cm);
  \draw (-4.5,0) circle (.2cm);
  \draw (-4.5-.2,0) arc (180:360:.2 and 0.06);
  \draw[dashed] (-4.5+.2,0) arc (0:180:.2 and 0.06);

\draw (-1.5,0) node {$[0,1,2]$};
\shade[ball color = red!40, opacity = 0.4] (-.5,0) circle (.2cm);
  \draw (-.5,0) circle (.2cm);
  \draw (-.5-.2,0) arc (180:360:.2 and 0.06);
  \draw[dashed] (-.5+.2,0) arc (0:180:.2 and 0.06);

\draw (2.5,0) node {$[0,2,3]$};
\shade[ball color = red!40, opacity = 0.4] (3.5,0) circle (.2cm);
  \draw (3.5,0) circle (.2cm);
  \draw (3.5-.2,0) arc (180:360:.2 and 0.06);
  \draw[dashed] (3.5+.2,0) arc (0:180:.2 and 0.06);

\draw (6,0) node {$[0,1,4]$};
\shade[ball color = red!40, opacity = 0.4] (7,0) circle (.2cm);
  \draw (7,0) circle (.2cm);
  \draw (7-.2,0) arc (180:360:.2 and 0.06);
  \draw[dashed] (7+.2,0) arc (0:180:.2 and 0.06);

\draw (-5.5,-2) node {$[1,3]$};
\draw (-5,-2) node[circ]{$$};
\draw (-3.5,-2) node {$[0,3]$ };
\shade[ball color = red!40, opacity = 0.4] (-2.9,-2) circle (.2cm);
  \draw (-2.9,-2) circle (.2cm);
  \draw (-2.9-.2,-2) arc (180:360:.2 and 0.06);
  \draw[dashed] (-2.9+.2,-2) arc (0:180:.2 and 0.06);

\draw (-1.5,-2) node {$[1,2]$};
\draw (-1,-2) node[circ]{$$};
\draw (.5,-2) node {$[0,1]$};
\shade[ball color = red!40, opacity = 0.4] (1.1,-2) circle (.2cm);
  \draw (1.1,-2) circle (.2cm);
  \draw (1.1-.2,-2) arc (180:360:.2 and 0.06);
  \draw[dashed] (1.1+.2,-2) arc (0:180:.2 and 0.06);
  \shade[ball color = red!40, opacity = 0.4] (1.5,-2) circle (.2cm);
  \draw (1.5,-2) circle (.2cm);
  \draw (1.5-.2,-2) arc (180:360:.2 and 0.06);
  \draw[dashed] (1.5+.2,-2) arc (0:180:.2 and 0.06);
  
\draw (2.5,-2) node {$[2,3]$};
\draw (3,-2) node[circ] {$$};
\draw (4.5,-2) node {$[0,2]$};
\shade[ball color = red!40, opacity = 0.4] (5.1,-2) circle (.2cm);
  \draw (5.1,-2) circle (.2cm);
  \draw (5.1-.2,-2) arc (180:360:.2 and 0.06);
  \draw[dashed] (5.1+.2,-2) arc (0:180:.2 and 0.06);
\draw (6.5,-2) node {$[1,4]$};
\draw (7,-2) node[circ]{$$};
\draw (8.5,-2) node {$[0,4]$};
\draw (9,-2) node[circ]{$$};

\draw (-4.5,-4) node {$[3]$};
\draw (-4.2,-4) node[circ]{$$};
\draw (-2,-4) node {$[2]$};
\draw (-1.7,-4) node[circ]{$$};
\draw (.5,-4) node {$[0]$};
\shade[ball color = red!40, opacity = 0.4] (1,-4) circle (.2cm);
  \draw (1,-4) circle (.2cm);
  \draw (1-.2,-4) arc (180:360:.2 and 0.06);
  \draw[dashed] (1+.2,-4) arc (0:180:.2 and 0.06);
\draw (3,-4) node {$[1]$};
\draw (3.3,-4) node[circ]{$$};
\draw (6.5,-4) node {$[4]$};
\draw (6.8,-4) node[circ]{$$};

\draw[blue,dashed]  (-5.5,-.2)--(-5.5,-1.8);
\draw[red]  (-5.5,-.2)--(-3.5,-1.8);
\draw[blue,dashed]  (-5.5,-.2)--(.5,-1.8);
\draw[blue,dashed]  (-1.5,-.2)--(-1.5,-1.8);
\draw[blue,dashed]  (-1.5,-.2)--(.5,-1.8);
\draw[red]  (-1.5,-.2)--(4.5,-1.8);
\draw[red] (2.5,-.2)--(-3.5,-1.8);
\draw[blue,dashed]   (2.5,-.2)--(2.5,-1.8);
\draw[red]  (2.5,-.2)--(4.5,-1.8);
\draw[blue,dashed]  (6,-.2)--(.5,-1.8);
\draw[blue,dashed]   (6,-.2)--(6.5,-1.8);
\draw[blue,dashed]  (6,-.2)--(8.5,-1.8);

\draw[red] (-5.5,-2.2)--(-4.5,-3.8);
\draw[red]  (-5.5,-2.2)--(3,-3.8);
\draw[blue,dashed]  (-3.5,-2.2)--(-4.5,-3.8);
\draw[red]  (-3.5,-2.2)--(.5,-3.8);
\draw[red]  (-1.5,-2.2)--(-2,-3.8);
\draw[red]  (-1.5,-2.2)--(3,-3.8);
\draw[blue,dashed]  (.5,-2.2)--(.5,-3.8);
\draw[blue,dashed]  (.5,-2.2)--(3,-3.8);
\draw[red]   (2.5,-2.2)--(-4.5,-3.8);
\draw[red]   (2.5,-2.2)--(-2,-3.8);
\draw[blue,dashed]  (4.5,-2.2)--(-2,-3.8);
\draw[red]   (4.5,-2.2)--(.5,-3.8);
\draw[red]   (6.5,-2.2)--(3,-3.8);
\draw[red] (6.5,-2.2)--(6.5,-3.8);
\draw[blue,dashed]  (8.5,-2.2)--(.5,-3.8);
\draw[red] (8.5,-2.2)--(6.5,-3.8);

\end{tikzpicture}
\end{center}
\caption{The Hasse diagram of the triangulated sundial. For any two adjacent simplices $\tau < \sigma$, an edge between $\tau$ and $\sigma$ in the diagram is solid red if $\Lcal(B_{\tau})\rightarrow \Lcal(B_{\sigma})$ is an isomorphism; otherwise it is in dotted blue. On the right of each simplex $\tau$ is either a point, a sphere, or the wedge of two spheres, chosen so that $\Lcal(B_\tau)$ is isomorphic to the reduced homology of the associated space.}
\label{fig:sundial-hasse}
\end{figure}
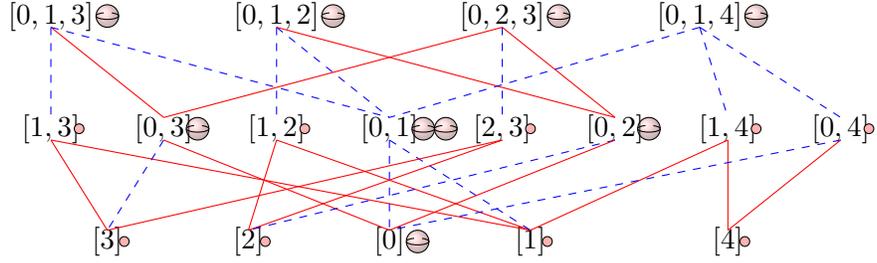

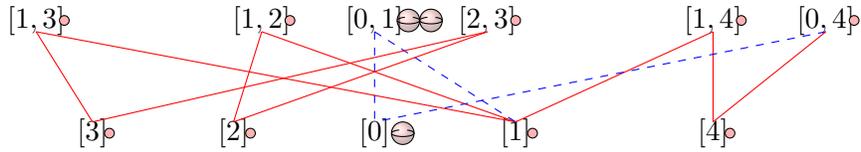
\begin{figure}[t!]
\begin{center}
\begin{tikzpicture}[scale=.75,circ/.style={
    circle,
    fill=red!30!white,
    draw,
    outer sep=0pt,
    inner sep=1.3pt
  }]

\draw (-5.5,-2) node {$[1,3]$};
\draw (-5,-2) node[circ]{$$};
\draw (-1.5,-2) node {$[1,2]$};
\draw (-1,-2) node[circ]{$$};
\draw (.5,-2) node {$[0,1]$};
\shade[ball color = red!40, opacity = 0.4] (1.1,-2) circle (.2cm);
  \draw (1.1,-2) circle (.2cm);
  \draw (1.1-.2,-2) arc (180:360:.2 and 0.06);
  \draw[dashed] (1.1+.2,-2) arc (0:180:.2 and 0.06);
  \shade[ball color = red!40, opacity = 0.4] (1.5,-2) circle (.2cm);
  \draw (1.5,-2) circle (.2cm);
  \draw (1.5-.2,-2) arc (180:360:.2 and 0.06);
  \draw[dashed] (1.5+.2,-2) arc (0:180:.2 and 0.06);
\draw (2.5,-2) node {$[2,3]$};
\draw (3,-2) node[circ]{$$};
\draw (6.5,-2) node {$[1,4]$};
\draw (7,-2) node[circ]{$$};
\draw (8.5,-2) node {$[0,4]$};
\draw (9,-2) node[circ]{$$};

\draw (-4.5,-4) node {$[3]$};
\draw (-4.2,-4) node[circ]{$$};
\draw (-2,-4) node {$[2]$};
\draw (-1.7,-4) node[circ]{$$};
\draw (.5,-4) node {$[0]$};
\shade[ball color = red!40, opacity = 0.4] (1,-4) circle (.2cm);
  \draw (1,-4) circle (.2cm);
  \draw (1-.2,-4) arc (180:360:.2 and 0.06);
  \draw[dashed] (1+.2,-4) arc (0:180:.2 and 0.06);
\draw (3,-4) node {$[1]$};
\draw (3.3,-4) node[circ]{$$};
\draw (6.5,-4) node {$[4]$};
\draw (6.8,-4) node[circ]{$$};

\draw[red] (-5.5,-2.2)--(-4.5,-3.8);
\draw[red]  (-5.5,-2.2)--(3,-3.8);
\draw[red]  (-1.5,-2.2)--(-2,-3.8);
\draw[red]  (-1.5,-2.2)--(3,-3.8);
\draw[blue,dashed]  (.5,-2.2)--(.5,-3.8);
\draw[blue,dashed]  (.5,-2.2)--(3,-3.8);
\draw[red]   (2.5,-2.2)--(-4.5,-3.8);
\draw[red]   (2.5,-2.2)--(-2,-3.8);
\draw[red]   (6.5,-2.2)--(3,-3.8);
\draw[red] (6.5,-2.2)--(6.5,-3.8);
\draw[blue,dashed]  (8.5,-2.2)--(.5,-3.8);
\draw[red] (8.5,-2.2)--(6.5,-3.8);

\end{tikzpicture}
\end{center}
\caption{The Hasse diagram after the top dimensional stratum has been removed. We can consider this the beginning of the second iteration of the algorithm in Section \ref{sec:clustering}.}
\label{fig:sundial-hasse2}
\end{figure}

A sheaf on $K$ can be considered as a labeling of each vertex in the Hasse diagram with a set and each edge with a morphism between the corresponding sets. 
Consider the local homology sheaf $\mathcal{L}$ on $K$ which takes each
open set \footnote{In the finite simplicial setting, $U$ is the support of a union of open simplices in $K$.} $U \subset K$ to $H_\bullet(K, K-U)\cong \tilde{H}_\bullet (K/(K-U))\cong \tilde{H}_\bullet (\closure(U)/\lk(U))\cong H_\bullet(\closure(U),\lk(U))$, where $\lk(U):=\closure(U)-U$. The local homology sheaf is naturally group-valued, but we make no use of the group structure here. So we will forget the group structure of local homology groups, and think of them purely as sets. The above isomorphisms follow from excision and the observation that $K-U$ (resp. $\lk(U)$) is a closed subcomplex of $K$ (resp. $\closure(U)$), and therefore ($K$,$K-U$) and ($\closure(U), \lk(U)$) form good pairs (see \cite[page 124]{Hatcher2000}).
Our algorithm described in Section~\ref{sec:clustering} can then be interpreted as computing local homology sheaf associated with each vertex in the Hasse diagram, and determining whether each edge in the diagram is an isomorphism. 
Our algorithm works by considering an element $\sigma$ in the Hasse diagram to be in the top-dimensional strata if all of the edges above $\sigma$ are isomorphisms, that is, if $\Lcal( \sigma<\tau)$ is an isomorphism for all pairs $ \sigma < \tau$. 

As illustrated in Figure~\ref{fig:sundial-hasse}, first, we start with the $2$-simplices. Automatically, we have that $\Lcal$ is constant when restricted to any $2$-simplex, and gives homology groups isomorphic to the reduced homology of a $2$-sphere. For instance, the local homology groups of the $2$-simplex $\sigma=[0,1,3]$ is isomorphic to the reduced homology of a 2-sphere, $H_\bullet(\overline{\str}(\sigma),\lk(\sigma)) \cong \tilde{H}_\bullet(\Sspace^2)$.  

Second, we consider the restriction of $\Lcal$ to the minimal open neighborhood of a $1$-simplex. 
For instance, consider the $1$-simplex $[1,3]$; 
$B_{[1,3]}=[1,3] \cup [0,1,3]$.
It can be seen that $\lk(B)_{[1,3]}=[0]\cup[3]\cup[1]\cup[0,3]\cup[0,1]$, and $H_\bullet(\closure(B_{[1,3]}),\lk(B_{[1,3]}))$ is isomorphic to the reduced homology of a single point space. Therefore the restriction map $\Lcal(B_{[1,3]})\rightarrow\Lcal(B_{[0,1,3]})$ is not an isomorphism (illustrated as a dotted blue line in Figure~\ref{fig:sundial-hasse}). 
On the other hand, let us consider the $1$-simplex $[0,3]$, where $B_{[0,3]}=[0,3]\cup[0,1,3]\cup [0,2,3]$. We have that 
$\lk(B_{[0,3]})=[0]\cup[1]\cup [2]\cup[3]\cup[0,1]\cup [0,2]\cup [1,3]\cup [2,3]\cup [1,3]$. 
Therefore $\Lcal(B_{[0,3]})$ is isomorphic to the reduced homology of a $2$-sphere. Moreover, both of the restriction maps corresponding to $B_{[0,1,3]}\subset B_{[0,3]}$ and $B_{[0,2,3]}\subset B_{[0,3]}$ are isomorphisms (illustrated as solid red lines in Figure~\ref{fig:sundial-hasse}). 
This implies that $[0,3] \in S_2 = X_2-X_1$. Alternatively, we can consider the simplex $[0,1]$ and see that $\Lcal(B_{[0,1]})\cong \tilde{H}_\bullet (\overline{\str}([0,1])/\lk([0,1]))\cong \tilde{H}_\bullet (S^2\vee S^2)$, the reduced homology of the wedge of two spheres. Therefore, for any 2-simplex $\tau$, the restriction map $\Lcal([0,1]<\tau)$ can not be an isomorphism. We conclude that $[0,1]$ is not contained in the top dimensional stratum.
If we continue, we see that the top dimensional stratum is given by 
$S_2 = [0,1,3]\cup[0,1,2]\cup[0,2,3]\cup[0,1,4]\cup[0,2]\cup[0,3]$, see Figure~\ref{fig:triangulated-sundial}.

Next, we can calculate the stratum $S_1 = X_1-X_0$ by only considering restriction maps whose codomain is not contained in $S_2$ (see Figure \ref{fig:sundial-hasse2}). We get 
$S_1=[0,1]\cup[1,3]\cup[1,2]\cup[2,3]\cup [0,4]\cup[1,4]\cup[2]\cup[3]\cup[4]$, which is visualized in Figure~\ref{fig:triangulated-sundial}. Finally, the stratum $S_0 = X_0$ consists of the vertices which have not been assigned to any strata. So $S_0=[0]\cup[1]$.

We think it prudent to point out several observations related to the above example. First, we see that the local homology groups assigned to a simplex are trivial if and only if the simplex belongs to the boundary of our space. In this sense, local homology can detect which simplices are on the boundary of the space without relying on any particular geometric realization of the abstract simplicial complex (embedding of the abstract simplicial complex into Euclidean space). Secondly, we observe that (for this example) the coarsest $\mathcal{L}$-stratification we calculated is actually the unique minimal homogeneous $\mathcal{L}$-stratification. We will investigate this coincidence for $\mathcal{L}$-stratifications elsewhere, in an attempt to say if a coarsest $\mathcal{L}$-stratification is automatically homogeneous or minimal.
For low-dimensional examples, we observe local homology based stratification we recover is actually a topological stratification. In general, local homology does not carry enough information to recover a stratification into manifold pieces, and examples exist in higher dimensions where $\mathcal{L}$-stratification are not topological stratifications.
One final remark related to the above example concerns the existence of restriction maps between isomorphic homology groups which are not isomorphisms. Suppose open sets $U\subset V$ in a finite $T_0$-space have the property that $\mathcal{L}(U)$ is isomorphic to $\mathcal{L}(V)$. A natural question to ask is if the restriction map $\mathcal{L}(U\subset V):\mathcal{L}(V)\rightarrow \mathcal{L}(U)$ is necessarily an isomorphism. This happens to be true for the sundial example above, but it is not true in general. See Figure~\ref{fig:pinchedtorus} for an illustration of a \emph{pinched torus}. The local homology of the pinched point and any point on the 1-dimensional strata are each isomorphic to the reduced homology of the wedge of two spheres. However, the restriction map from an open neighborhood of the pinched point to an open neighborhood of a 1-simplex in the one strata which is adjacent to the pinched point is not an isomorphism. The coarsest $\Lcal$-stratification we obtain from our algorithm coincides with the stratification given in Figure \ref{fig:pinchedtorus} (for a suitable triangulation of the pinched torus).


\section{Stratification Learning with Sheaf of Maximal Elements}
\label{sec:maximal-element}

We will now consider a stratification of the triangulated sundial given by the sheaf of maximal elements (defined below). Again, let $|K|$ be the polyhedron in Figure~\ref{fig:triangulated-sundial} with labeled vertices. 

\para{Sheaf of maximal elements.}
Consider the sheaf $\Fcal$ on the space $K$ which takes each open set $U\subset K$ to the free $\mathbb{Z}$-module generated by maximal elements of $U$. For $V\subset U$, $\Fcal(U)\rightarrow \Fcal(V)$ maps an element $u\in U$ to $u\in V$ if $u\in V$ and 0 otherwise. 

\para{Stratification learning using sheaf of maximal elements.}
Now for the triangulated sundial, $\Fcal$ is automatically constant when restricted to any of the 2-simplices. 

Let us consider the restriction of $\Fcal$ to the minimal open set containing $[1,3]$, that is, $B_{[1,3]}=[1,3]\cup [0,1,3]$.
The restriction map $\Fcal(B_{[1,3]})\rightarrow \Fcal(B_{[0,1,3]})$ sends the only maximal element in $B_{[1,3]}$ to the only maximal element in $B_{[0,1,3]}$, and therefore is an isomorphism from $\mathbb{Z}$ to $\mathbb{Z}$. 

Let us consider a more subtle example. The minimal open set containing $[0,3]$ is given by $B_{[0,3]}=[0,3]\cup[0,1,3]\cup [0,2,3]$. We see that there are two distinct maximal elements, $[0,1,3]$ and $ [0,2,3]$. Therefore $\Fcal(B_{[0,3]})=\mathbb{Z}^2$. This means that neither of the restriction maps can be isomorphisms, since $\Fcal(B_{[0,1,3]})$ and $\Fcal(B_{[0,2,3]})$ are each isomorphic to $\mathbb{Z}$. 

If we continue, we see that the top stratum is given by 
$$
S_2 = X_2-X_1=[0,1,3]\cup[0,1,2]\cup[0,2,3]\cup[0,1,4]\cup[1,3]\cup[1,2]\cup[2,3]\cup[1,4]\cup[0,4].
$$
For this example, the stratum $S_2$ is automatically homogeneous, hence the two algorithm variations in Section \ref{sec:clustering} (2a and 2a') produce the same 2-stratum $S_2$, illustrated in Figure~\ref{fig:triangulated-sundial-maximal-chain}. We construct the labeled Hasse diagram based on the sheaf of maximal elements, as shown in Figure~\ref{fig:sundial-hasse-maximal-chain}. 

\begin{figure}[ht!]

\begin{center}
\begin{tikzpicture}[scale=.8]
\draw (-5.5,0) node {$[0,1,3]$};
\draw (-1.5,0) node {$[0,1,2]$};
\draw (2.5,0) node {$[0,2,3]$};
\draw (6,0) node {$[0,1,4]$};

\draw (-5.5,-2) node {$[1,3]$};
\draw (-3.5,-2) node {$[0,3]$};
\draw (-1.5,-2) node {$[1,2]$};
\draw (.5,-2) node {$[0,1]$};
\draw (2.5,-2) node {$[2,3]$};
\draw (4.5,-2) node {$[0,2]$};
\draw (6.5,-2) node {$[1,4]$};
\draw (8.5,-2) node {$[0,4]$};

\draw (-4.5,-4) node {$[3]$};
\draw (-2,-4) node {$[2]$};
\draw (.5,-4) node {$[0]$};
\draw (3,-4) node {$[1]$};
\draw (6.5,-4) node {$[4]$};

\draw[red] (-5.5,-.2)--(-5.5,-1.8);
\draw[blue,dashed] (-5.5,-.2)--(-3.5,-1.8);
\draw[blue,dashed]  (-5.5,-.2)--(.5,-1.8);
\draw[red] (-1.5,-.2)--(-1.5,-1.8);
\draw[blue,dashed]  (-1.5,-.2)--(.5,-1.8);
\draw[blue,dashed]  (-1.5,-.2)--(4.5,-1.8);
\draw[blue,dashed]  (2.5,-.2)--(-3.5,-1.8);
\draw[red] (2.5,-.2)--(2.5,-1.8);
\draw[blue,dashed]  (2.5,-.2)--(4.5,-1.8);
\draw[blue,dashed]  (6,-.2)--(.5,-1.8);
\draw[red] (6,-.2)--(6.5,-1.8);
\draw[red] (6,-.2)--(8.5,-1.8);

\draw[blue,dashed] (-5.5,-2.2)--(-4.5,-3.8);
\draw[blue,dashed]  (-5.5,-2.2)--(3,-3.8);
\draw[red]  (-3.5,-2.2)--(-4.5,-3.8);
\draw[blue,dashed]  (-3.5,-2.2)--(.5,-3.8);
\draw[blue,dashed]  (-1.5,-2.2)--(-2,-3.8);
\draw[blue,dashed]  (-1.5,-2.2)--(3,-3.8);
\draw[blue,dashed]  (.5,-2.2)--(.5,-3.8);
\draw[red]  (.5,-2.2)--(3,-3.8);
\draw[blue,dashed]  (2.5,-2.2)--(-2,-3.8);
\draw[blue,dashed]  (2.5,-2.2)--(-4.5,-3.8);
\draw[blue,dashed]  (4.5,-2.2)--(.5,-3.8);
\draw[red]  (4.5,-2.2)--(-2,-3.8);
\draw[blue,dashed]  (6.5,-2.2)--(3,-3.8);
\draw[red] (6.5,-2.2)--(6.5,-3.8);
\draw[blue,dashed]  (8.5,-2.2)--(.5,-3.8);
\draw[red] (8.5,-2.2)--(6.5,-3.8);
\end{tikzpicture}
\end{center}
\caption{The labeled Hasse diagram for the sundial with the sheaf of maximal elements. Red solid lines denote isomorphisms; while blue dotted lines denote non-isomorphisms. }
\label{fig:sundial-hasse-maximal-chain}
\end{figure}
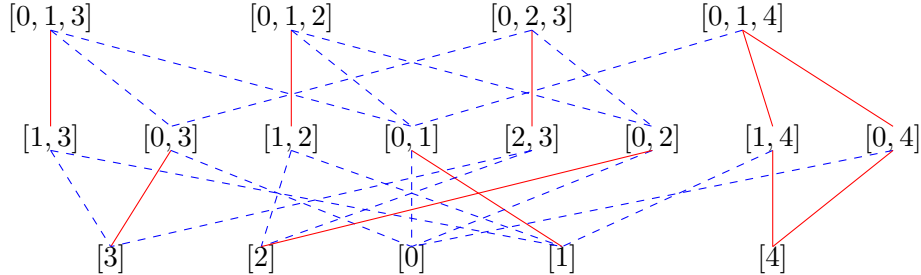

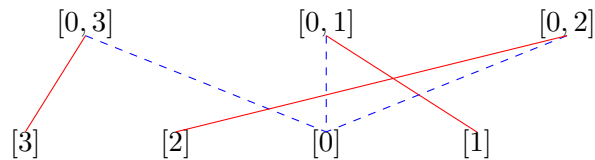
\begin{figure}[ht!]

\begin{center}
\begin{tikzpicture}[scale=.8]

\draw (-3.5,-2) node {$[0,3]$};

\draw (.5,-2) node {$[0,1]$};

\draw (4.5,-2) node {$[0,2]$};

\draw (-4.5,-4) node {$[3]$};
\draw (-2,-4) node {$[2]$};
\draw (.5,-4) node {$[0]$};
\draw (3,-4) node {$[1]$};

\draw[red]  (-3.5,-2.2)--(-4.5,-3.8);
\draw[blue,dashed]  (-3.5,-2.2)--(.5,-3.8);

\draw[blue,dashed]  (.5,-2.2)--(.5,-3.8);
\draw[red]  (.5,-2.2)--(3,-3.8);

\draw[blue,dashed]  (4.5,-2.2)--(.5,-3.8);
\draw[red]  (4.5,-2.2)--(-2,-3.8);

\end{tikzpicture}
\end{center}
\caption{The labeled Hasse diagram for the sundial with the sheaf of maximal elements after the top dimensional stratum has been removed. }
\label{fig:sundial-hasse-maximal-elmt2}
\end{figure}

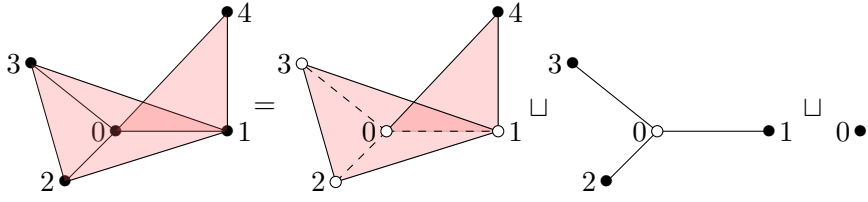
\begin{figure}[ht!]
\begin{center}
\begin{tikzpicture}[scale=.9,circ/.style={
    circle,
    fill=white,
    draw,
    outer sep=0pt,
    inner sep=1.5pt
  }]
\draw (0,0) node {$\bullet$};
\draw (0,0) node[left] {0};
\draw (1.65,0) node {$\bullet$};
\draw (1.65,0) node[right] {1};
\draw (1.65,1.75) node {$\bullet$};
\draw (1.65,1.75) node[right] {4};
\draw (-1.25,1) node {$\bullet$};
\draw (-1.25,1) node[left] {3};
\draw (-.75,-.75) node {$\bullet$};
\draw (-.75,-.75) node[left] {2};
\draw (0,0)--(1.65,0);
\draw (0,0)--(1.65,1.75);
\draw (1.65,0)--(1.65,1.75);
\draw (1.65,0)--(-1.25,1);
\draw (1.65,0)--(-.75,-.75);
\draw (-1.25,1)--(-.75,-.75);
\draw (0,0)--(-.75,-.75);
\draw (0,0)--(-1.25,1);
\draw[draw=black,fill=red!60!white, nearly transparent] (0,0)--(1.65,0)--(1.65,1.75);
\draw[draw=black,fill=red!60!white, nearly transparent] (-.75,-.75)--(1.65,0)--(-1.25,1);

\draw (2.2,.35) node{$=$};
\draw[draw=black,fill=red!60!white, nearly transparent] (4+0,0)--(4+1.65,0)--(4+1.65,1.75);
\draw[draw=black,fill=red!60!white, nearly transparent] (4+-.75,-.75)--(4+1.65,0)--(4+-1.25,1);
\draw (4+0,0) node[left] {0};
\draw (4+1.65,0) node[right] {1};
\draw (4+1.65,1.75) node[right] {4};
\draw (4+-1.25,1) node[left] {3};
\draw (4+-.75,-.75) node[left] {2};

\draw[dashed] (4+0,0)--(4+1.65,0);
\draw[]  (4+0,0)--(4+1.65,1.75);
\draw[] (4+1.65,0)--(4+1.65,1.75);
\draw[] (4+1.65,0)--(4-1.25,1);
\draw[]  (4+1.65,0)--(4+-.75,-.75);
\draw[]  (4+-1.25,1)--(4+-.75,-.75);
\draw[dashed] (4+0,0)--(4+-.75,-.75);
\draw[dashed] (4+0,0)--(4+-1.25,1);

\draw (4+-.75,-.75) node[circ] {};
\draw (4+-1.25,1) node[circ] {};
\draw (4+1.65,1.75) node[] {$\bullet$};
\draw (4+1.65,0) node[circ] {};
\draw (4+0,0) node[circ]{};

\draw (6.25,.35) node{$\sqcup$};

\draw (8+0,0) node[left] {0};
\draw (8+1.65,0) node[right] {1};
\draw (8+-1.25,1) node[left] {3};
\draw (8+-.75,-.75) node[left] {2};
\draw[] (8+0,0)--(8+1.65,0);
\draw[] (8+0,0)--(8+-.75,-.75);
\draw[] (8+0,0)--(8+-1.25,1);

\draw (8+-.75,-.75) node[] {$\bullet$};
\draw (8+-1.25,1) node[] {$\bullet$};
\draw (8+1.65,0) node[] {$\bullet$};
\draw (8+0,0) node[circ]{};

\draw (10.3,.35) node{$\sqcup$};

\draw (11,0) node[left] {0};
\draw (11,0) node[]{$\bullet$};

\end{tikzpicture}
\caption{A triangulated sundial and its stratification based on the sheaf of maximal elements.}
\label{fig:triangulated-sundial-maximal-chain}
\end{center}
\end{figure}

Next, we can calculate the strata $S_1$ by only considering restriction maps whose codomain is not contained in $S_2$ (see Figure \ref{fig:sundial-hasse-maximal-elmt2}). We get 
$$
 S_1 = X_1-X_0 = [0,1]\cup[0,2]\cup [0,3]\cup[1]\cup[2]\cup [3]
$$ 
 We can consider the Hasse diagram and corresponding visualization of $S_1$, as illustrated in Figure~\ref{fig:triangulated-sundial-maximal-chain}. Again, the stratum $S_1$ is automatically homogeneous, meaning that the algorithm variations 2a and 2a' produce the same output. 

Finally, the strata $X_0$ in the coarsest $\Fcal$-stratification consists of the points which have not been assigned to a strata yet. So 
$$S_0 = X_0=[0].$$ 
Intuitively, we are using this relatively simple sheaf to cluster the space $|K|$ into stratum pieces where small neighborhoods of points in the same stratum piece intersect the same set of 2-simplices. 

\section{Stratification Learning Using Geometric Techniques}
\label{sec:geometric-sheaf}

\subsection{Pre-Sheaf of Vanishing Polynomials}
In this section we will use \emph{Learning Algebraic Varieties from Samples}~\cite{BreidingKalisnikSturmfels2018} to stratify the nerve of an open cover of a point cloud data set. Suppose $X\subset \mathbb{R}^n$ is a finite set of points, and $\{U_i\}$ is a finite cover of $X$, such that $U_i\subset X$ for each $i$, and $X=\bigcup_iU_i$. We will proceed by outlining a geometric method for computing a stratification of the nerve of $\{U_i\}$, $\mathcal{N}$, viewed as a finite $T_0$-space.

We will begin by briefly reviewing the approach to learning algebraic varieties described in~\cite{BreidingKalisnikSturmfels2018}. To each algebraic set $S\subset\mathbb{R}^n$, defined to be the set of solutions to a system of polynomial equations, we can associate an ideal of polynomial functions 
$$I(S):=\{ \text{ polynomial function  } p \text{ on }\mathbb{R}^n: p(x)=0\quad\forall x\in S\}.$$
If $\Omega$ is a finite set of points sampled from $S$, then $I(\Omega)\supset I(S)$. The insight used in~\cite{BreidingKalisnikSturmfels2018}, is that certain finite dimensional subspaces of polynomial functions will not be able to distinguish $\Omega$ from $S$. More precisely, we will start with a finite set of linearly independent polynomial functions $\mathcal{M}$, and consider the subspace $R_\mathcal{M}$ consisting of $\mathbb{R}$-linear combinations of elements in $\mathcal{M}$. To a given set $V\subset\mathbb{R}^n$, we will associate the subspace of polynomial functions in $R_\mathcal{M}$ which vanish on $V$:
$$
I_\mathcal{M}(V):=\{p\in R_\mathcal{M}:p(x)=0\quad\forall x\in V\}.
$$
The goal is to carefully choose a finite set of polynomials $\mathcal{M}$ so that $I_\mathcal{M}(\Omega)=I_\mathcal{M}(S)$ and $I_\mathcal{M}(S)$ generates the ideal $I(S)$ in the ring of polynomial functions. 

The pre-sheaf of vanishing polynomials $\mathcal{I}_\mathcal{M}$ is defined using the above association of a finite dimensional vector space $I_\mathcal{M}(V)$ to various point sets $V\subset\mathbb{R}^n$. Given a collection of subsets $\{U_i\}_{i\in J}$, we define an abstract simplicial complex on the index set $J$ by declaring a subset $\tau\subseteq J$ to be a simplex if and only if the corresponding intersection
$$
V_\tau :=\bigcap_{i\in \tau}U_i\subset X\subset\mathbb{R}^n
$$
is non-empty. Finally, given an open set $W\subset \mathcal{N}$ (with respect to the Alexandroff topology defined in Section \ref{subsec:t0}), define 
$$
X_W=\bigcup_{\tau\in W}V_\tau.
$$
The pre-sheaf of vanishing polynomials $\mathcal{I}_\mathcal{M}$ is defined by
$$
\mathcal{I}_\mathcal{M}(W):=I_\mathcal{M}(X_W).
$$
{\bf{Remark.}} We could sheafify this pre-sheaf, and get a sheaf of vanishing, locally polynomial functions. However, since the algorithm outlined in this paper only requires a pre-sheaf as an input, we will continue without taking into account the larger space of locally polynomial functions.
\begin{proposition}
For every finite set of linearly independent polynomials $\mathcal{M}$, the contravariant functor $\mathcal{I}_\mathcal{M}$ from the category of open sets of $\mathcal{N}$ to the category of finite dimensional vector spaces is a pre-sheaf. 
\end{proposition}
\begin{proof}
Assume $W\subset Y\subset \mathcal{N}$ are two open sets. Then $X_W\subset X_Y$. If $p\in I_\mathcal{M}(X_Y)$, then by definition $p$ vanishes on the set $X_W$. Therefore $I_\mathcal{M}(X_Y)\subset I_\mathcal{M}(X_W)$. The restriction map induced by $\mathcal{I}_\mathcal{M}$ is the inclusion 
\begin{eqnarray*}
\mathcal{I}_\mathcal{M}(W\subset Y):I_\mathcal{M}(X_Y)&\longhookrightarrow& I_\mathcal{M}(X_W)\\
f&\mapsto& f
\end{eqnarray*}
To see that this map is well defined, we notice that if $f$ vanishes on $X_Y$, and if $X_W\subset X_Y$, then $f$ must vanish on $X_W$. It follows that $\mathcal{I}_\mathcal{M}(U\subset U)=\text{id}_U$ and $\mathcal{I}_\mathcal{M}(U\subset W)\circ\mathcal{I}_\mathcal{M}(W\subset Y)=\mathcal{I}_\mathcal{M}(U\subset Y)$. 

\end{proof}
\noindent{\bf{Remark.}} Since the restriction maps $\mathcal{I}_\mathcal{M}(W\subset Y)$ are necessarily injective, we can conclude that if $\mathcal{I}_\mathcal{M}(W)$ is isomorphic to $\mathcal{I}_\mathcal{M}(Y)$, then the restriction map $\mathcal{I}_\mathcal{M}(W\subset Y)$ is an isomorphism. Therefore, computing the $\mathcal{I}_\mathcal{M}$-stratification reduces to computing the stalks $\mathcal{I}_\mathcal{M}(\text{St}\tau)$ for each simplex $\tau$ (rather than computing the restriction maps).  

\subsection{Examples}
Here we will illustrate examples of geometric stratifications of open covers of finite point sets in $\mathbb{R}^2$ and $\mathbb{R}^3$. These examples aim to illustrate some of the features, as well as subtleties, of the $\mathcal{I}_\mathcal{M}$-stratifications described above.  

\para{The circle $S^1$.} 
 Intuitively, we expect that the $\mathcal{I}_\mathcal{M}$-stratification will be trivial (i.e., will consist of a single stratum) when our underlying geometric space is sufficiently well behaved (an analytic manifold, for example). We will begin by checking this intuition for the $\mathcal{I}_\mathcal{M}$-stratification of a circle. 
Consider the finite set $\Omega$ consisting of 100 points on the unit circle $S^1$ in $\mathbb{R}^2$ spaced at regular intervals, with an open cover $\mathcal{U}_\Omega=\{U_1,\cdots, U_6\}$ consisting of six open sets (see Figure \ref{fig:circle}).
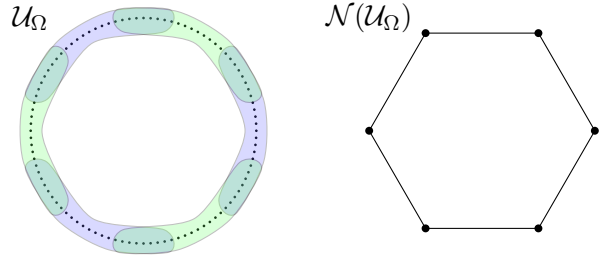
\begin{figure}[ht!]
\begin{center}
\begin{tikzpicture}[scale=1.5,circ/.style={
    circle,
    fill=white,
    draw,
    outer sep=0pt,
    inner sep=1.5pt
  }]
  \draw (-1,1) node {$\mathcal{U}_\Omega$};
  \draw (2,1) node {$\mathcal{N}(\mathcal{U}_\Omega)$};
\foreach \i in {1,...,100}
{ 
\node[circle,fill=black,inner sep=0pt,minimum size=1pt] (a) at ({cos(3.6*\i)},{sin(3.6*\i)}) {};}
\foreach \i in {2,4,6}
{ 
\draw[draw=black,fill=blue!60!white, nearly transparent] plot [smooth cycle] coordinates {({.9*cos(60*(\i-.7))},{.9*sin(60*(\i-.7))})({1.1*cos(60*(\i-.7))},{1.1*sin(60*(\i-.7))})({1.1*cos(60*(\i-.23))},{1.1*sin(60*(\i-.23))})({1.1*cos(60*(\i+.23))},{1.1*sin(60*(\i+.23))})({1.1*cos(60*(\i+.7))},{1.1*sin(60*(\i+.7))})({.9*cos(60*(\i+.7))},{.9*sin(60*(\i+.7))})({.9*cos(60*(\i))},{.9*sin(60*(\i))})};}
\foreach \i in {1,3,5}
{ 
\draw[draw=black,fill=green!60!white, nearly transparent] plot [smooth cycle] coordinates {({.9*cos(60*(\i-.7))},{.9*sin(60*(\i-.7))})({1.1*cos(60*(\i-.7))},{1.1*sin(60*(\i-.7))})({1.1*cos(60*(\i-.23))},{1.1*sin(60*(\i-.23))})({1.1*cos(60*(\i+.23))},{1.1*sin(60*(\i+.23))})({1.1*cos(60*(\i+.7))},{1.1*sin(60*(\i+.7))})({.9*cos(60*(\i+.7))},{.9*sin(60*(\i+.7))})({.9*cos(60*(\i))},{.9*sin(60*(\i))})};}
\foreach \i in {1,...,6}
{ 
\node[circle,fill=black,inner sep=0pt,minimum size=3pt] (a) at ({cos(60*\i)+3},{sin(60*\i)}) {};
\draw ({cos(60*\i)+3},{sin(60*\i)})--({cos(60*(\i+1))+3},{sin(60*(\i+1))});}
\end{tikzpicture}
\caption{An open cover $\mathcal{U}_\Omega$ of $\Omega$, with corresponding nerve $\mathcal{N}(\mathcal{U}_\Omega)$.} 
\label{fig:circle}
\end{center}
\end{figure}

Suppose $\mathcal{M} = \{1,x,y,x^2,y^2,xy\}$ and $U=\cap_{i\in I} U_i$ is a subset of $\Omega$ corresponding to a simplex of $\mathcal{N}(\mathcal{U}_\Omega)$. 
The subspace of polynomials in $R_\mathcal{M}$ which vanish on $U$ is equal to the kernel of the linear map 
\begin{eqnarray*}
R_\mathcal{M}&\rightarrow& \mathbb{R}^{|U|}\\
p&\mapsto& (p(x_1,y_1),\cdots,p(x_k,y_k))
\end{eqnarray*}
where $(x_i,y_i)$ are the elements of $U$ with some fixed order. See \cite[Section 5]{BreidingKalisnikSturmfels2018} for a discussion concerning how to optimize our choice of $\mathcal{M}$, possibly resulting in the above map being represented by a sparse matrix. For our example we will assume that the subspace of polynomials in $R_\mathcal{M}$ which vanish on $U$ is the \emph{$\mathbb{R}$-span} of $x^2+y^2-1$:
$$
I_\mathcal{M}(U)=\mathbb{R}\langle x^2+y^2-1\rangle=\{r(x^2+y^2-1):r\in\mathbb{R}\}. 
$$
Notice that this calculation follows for any open set $U\subset \mathcal{N}(\mathcal{U}_\Omega)$. If $V$ and $U$ are two such sets with $V\subset U$, then 
$$
\mathcal{I}_\mathcal{M}(V\subset U):I_\mathcal{M}(U)\rightarrow I_\mathcal{M}(V)
$$
is an isomorphism. Therefore, our pre-sheaf $\mathcal{I}_\mathcal{M}$ is constant. This implies that the $\mathcal{I}_\mathcal{M}$-stratification of $\mathcal{N}(\mathcal{U}_\Omega)$ is the trivial stratification (i.e., the stratification of $\mathcal{N}(\mathcal{U}_\Omega)$ consisting of a single stratum). It should be noted that for degree reasons, the trivial stratification would be produced even if fewer points were sampled from the space. Specifically, as long as each cover element $U_i$ contains at least four distinct points, then the resulting presheaf $I_\mathcal{M}$ will be constant, and the $I_\mathcal{M}$-stratification will be trivial. 
\begin{figure}[ht!]
\begin{center}
\begin{tikzpicture}[scale=1.3,circ/.style={
    circle,
    fill=white,
    draw,
    outer sep=0pt,
    inner sep=1.5pt
  }]
  
  \draw (2,1) node {$\mathcal{N}(\mathcal{U}_\Omega)$};

\foreach \i in {1,...,6}
{ 
\node[circle,fill=black,inner sep=0pt,minimum size=3pt] (a) at ({cos(60*\i)+3},{sin(60*\i)}) {};
\draw ({cos(60*\i)+3},{sin(60*\i)})--({cos(60*(\i+1))+3},{sin(60*(\i+1))});}
\begin{scope}[shift={(3,0)}]
 
\draw (1.5,0) node {$=$};
\foreach \i in {1,...,6}
{ 
\node[circle,fill=black,inner sep=0pt,minimum size=3pt] (a) at ({cos(60*\i)+3},{sin(60*\i)}) {};
\draw ({cos(60*\i)+3},{sin(60*\i)})--({cos(60*(\i+1))+3},{sin(60*(\i+1))});}
\end{scope}
\end{tikzpicture}
\caption{The minimal homogeneous $\mathcal{I}_\mathcal{M}$-stratification of $\mathcal{N}(\mathcal{U}_\Omega)$ is trivial.} 
\label{fig:circle-poly}
\end{center}
\end{figure}
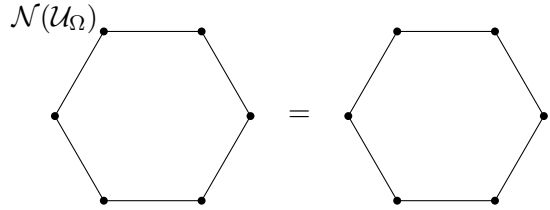

\para{A corner.} 
As a second example, we want to explore the extent to which the geometrically defined $\mathcal{I}_\mathcal{M}$-stratification is capable of detecting geometric singularities. Such a feature will provide an illustration of the key differences between $\mathcal{I}_\mathcal{M}$-stratifications and local homology stratifications. Consider the finite set $\Omega:=\{(0.1n,0):\text{ $n=0,\cdots,20$}\}\cup \{(0,0.1n):\text{ $n=0,\cdots,20$}\} \subset\mathbb{R}^2$, with the open cover $\mathcal{U}_\Omega$ depicted in Figure \ref{fig:axis}.  
Suppose (as in the previous example) that $\mathcal{M}=\{1,x,y,x^2,y^2,xy\}$. If $W\in \mathcal{U}_\Omega$ consists of elements of the form $(x,0)$, where $x\neq 0 $, then $\mathcal{I}_\mathcal{M}(W)=\mathbb{R}\langle y\rangle \oplus \mathbb{R}\langle y^2\rangle \oplus \mathbb{R}\langle xy\rangle $. Similarly, if $U\in \mathcal{U}_\Omega$ consists of elements of the form $(0,y)$, where $y\neq 0 $, then $\mathcal{I}_\mathcal{M}(U)=\mathbb{R}\langle x\rangle\oplus \mathbb{R}\langle x^2\rangle \oplus \mathbb{R}\langle xy\rangle  $. If $V\subset\mathcal{U}_\Omega$ contains $(0,0)$, then $\mathcal{I}_\mathcal{M}(V)=\mathbb{R}\langle xy\rangle$. The restriction map of vector spaces 
$$
\mathcal{I}(U\subset V):\mathbb{R}\langle xy\rangle \longhookrightarrow \mathbb{R}\langle x\rangle\oplus \mathbb{R}\langle x^2\rangle \oplus \mathbb{R}\langle xy\rangle  
$$
is not an isomorphism. The resulting $\mathcal{I}_\mathcal{M}$-stratification of $\mathcal{N}(\mathcal{U}_\Omega)$ is illustrated in Figure \ref{fig:axis}.
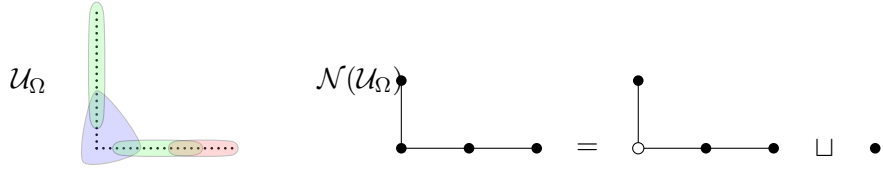
\begin{figure}[ht!]
\begin{center}
\begin{tikzpicture}[scale=.9,circ/.style={
    circle,
    fill=white,
    draw,
    outer sep=0pt,
    inner sep=1.5pt
  }]
    \draw (-1,1) node {$\mathcal{U}_\Omega$};
  \draw (3.9,1) node {$\mathcal{N}(\mathcal{U}_\Omega)$};
\foreach \i in {1,...,21}
{ 
\node[circle,fill=black,inner sep=0pt,minimum size=1pt] (a) at ({.1*\i-.1},{0}) {};
\node[circle,fill=black,inner sep=0pt,minimum size=1pt] (a) at ({0},{.1*\i-.1}) {};}

\draw[draw=black,fill=blue!60!white, nearly transparent] plot [smooth cycle] coordinates {(.65,0)(0,.85)(-.2,-.2)};
\draw[draw=black,fill=green!60!white, nearly transparent] plot [smooth cycle] coordinates {(.35,-.1)(.35,.1)(1.45,.1)(1.45,-.1)};
\draw[draw=black,fill=red!60!white, nearly transparent] plot [smooth cycle] coordinates {(1.15,-.1)(1.15,.1)(2,.1)(2,-.1)};
\draw[draw=black,fill=green!60!white, nearly transparent] plot [smooth cycle] coordinates {(-.1,.45)(.1,.45)(.1,2)(-.1,2)};
\begin{scope}[shift={(4.5,0)}]

\node[circle,fill=black,inner sep=0pt,minimum size=4pt] at (0,0) {};
\node[circle,fill=black,inner sep=0pt,minimum size=4pt] at (1,0) {};
\node[circle,fill=black,inner sep=0pt,minimum size=4pt] at (0,1) {};
\node[circle,fill=black,inner sep=0pt,minimum size=4pt] at (2,0) {};
\draw (0,0)--(2,0);
\draw (0,0)--(0,1);

\draw (2.75,0) node{$=$};
\begin{scope}[shift={(3.5,0)}]
\draw (0,0)--(2,0);
\draw (0,0)--(0,1);
\node[circ] at (0,0) {};
\node[circle,fill=black,inner sep=0pt,minimum size=4pt] at (1,0) {};
\node[circle,fill=black,inner sep=0pt,minimum size=4pt] at (0,1) {};
\node[circle,fill=black,inner sep=0pt,minimum size=4pt] at (2,0) {};
\end{scope}

\draw (6.25,0) node{$\sqcup$};

\node[circle,fill=black,inner sep=0pt,minimum size=4pt] at (7,0) {};
\end{scope}
\end{tikzpicture}
\caption{A stratification of $\mathcal{N}(\mathcal{U}_\Omega)$ using the pre-sheaf of vanishing polynomials.} 
\label{fig:axis}
\end{center}
\end{figure}

\para{The curve $y^2=x^3+x^2$.} 
Now we will give an example that aims to illustrate some of the more subtle properties of $\mathcal{I}_\mathcal{M}$-stratifications. Specifically, we will see a singularity that the local homology sheaf would detect, but the $\mathcal{I}_\mathcal{M}$ pre-sheaf does not.  
This example will consist of a finite set of solutions of $y^2=x^3+x^2$. Suppose $\mathcal{M}=\{1,x,y,x^2,y^2,xy,x^3,x^2y, xy^2,y^3\}$. The set of all real solutions of $y^2=x^3+x^2$, denoted by $X$, is parametrized by the map 
\begin{eqnarray*}
\phi:\mathbb{R}&\rightarrow& X\subset \mathbb{R}^2\\
t&\mapsto& \left(t^2-1,t^3-t\right).
\end{eqnarray*}
Suppose $f\in \mathcal{I}_\mathcal{M}(U)$, for a given open set $U\subset X$ (in the subspace topology). The function 
$$
g(t)=f(\phi(t))\in \mathbb{R}[t]
$$
is a polynomial function from $\mathbb{R}$ to $\mathbb{R}$. Let $V=\phi^{-1}(U)$. Since $f$ vanishes on $U$, $g$ must vanish on $V$. Since $V$ is an open subset of $\mathbb{R}$, we can conclude that $g$ is the zero polynomial. Therefore $f$ vanishes everywhere on $X$. Therefore, the $\mathcal{I}_\mathcal{M}$-stratification of $X$ is the trivial stratification, even though $X$ has a singular point at $(0,0)$.

Informed by the calculation above, we turn our attention to finite open cover of an $\epsilon$-net of $X$ (illustrated in Figure \ref{fig:elliptic}). For each $U\in\mathcal{U}_\Omega$, the vector space of vanishing polynomials is 
$$
\mathcal{I}_\mathcal{M}(U)= \mathbb{R}\langle x^3+x^2-y^2\rangle. 
$$
Moreover, the pre-sheaf $\mathcal{I}_\mathcal{M}$ is constant, and the resulting $\mathcal{I}_\mathcal{M}$-stratification is the trivial stratification, illustrated in Figure \ref{fig:elliptic}. Notice that for this example (as well as the circle example) the presheaf $\mathcal{I}_\mathcal{M}$ would be constant for most choices of cover $\mathcal{U}_\Omega$ of $X$. Specifically, if $\mathcal{U}_\Omega$ consists of two intersecting sets, then the corresponding nerve $\mathcal{N}(\mathcal{U}_\Omega)$ will consist of two vertices connected by an edge. The $\mathcal{I}_\mathcal{M}$-stratification of the 1-simplex $\mathcal{N}(\mathcal{U}_\Omega)$ will be the trivial stratification. 

\begin{figure}[ht!]
\begin{center}
\begin{tikzpicture}[scale=1.2,circ/.style={
    circle,
    fill=white,
    draw,
    outer sep=0pt,
    inner sep=1.5pt
  }]
    \draw (-1,1) node {$\mathcal{U}_\Omega$};
  \draw (2,1) node {$\mathcal{N}(\mathcal{U}_\Omega)$};
\foreach \i in {-26,...,26}
{ 
\node[circle,fill=black,inner sep=0pt,minimum size=1pt] (a) at ({2*(.03*\i)*(.03*\i)-2},{2*(.03*\i)*(.03*\i)*(.03*\i) -2*(.03*\i)}) {};}
\foreach \i in {-20,...,17}
{ 
\node[circle,fill=black,inner sep=0pt,minimum size=1pt] (a) at ({2*(.01*\i+1)*(.01*\i+1)-2},{2*(.01*\i+1)*(.01*\i+1)*(.01*\i+1) -2*(.01*\i+1)}) {};
\node[circle,fill=black,inner sep=0pt,minimum size=1pt] (a) at ({2*(-.01*\i-1)*(-.01*\i-1)-2},{2*(-.01*\i-1)*(-.01*\i-1)*(-.01*\i-1) -2*(-.01*\i-1)}) {};}

\draw[draw=black,fill=blue!60!white, nearly transparent] plot [smooth cycle] coordinates {(.5,.5)(0,.25)(-.5,.45)(-.25,0)(-.5,-.45)(0,-.25)(.5,-.5)(.25,0)};
\draw[draw=black,fill=green!60!white, nearly transparent] plot [smooth cycle] coordinates {(-.2,.1)(-.2,.3)(-1.25,.9)(-2,.5)(-2,.3)(-1.2,.6)};
\draw[draw=black,fill=green!60!white, nearly transparent] plot [smooth cycle] coordinates {(-.2,-.1)(-.2,-.3)(-1.25,-.9)(-2,-.5)(-2,-.3)(-1.2,-.6)};
\draw[draw=black,fill=red!60!white, nearly transparent] plot [smooth cycle] coordinates {(-1.7,.4)(-1.9,.6)(-2.2,0)(-1.9,-.6)(-1.7,-.4)(-1.8,0)};
\draw[draw=black,fill=green!60!white, nearly transparent] plot [smooth cycle] coordinates {(.2,.1)(.1,.3)(.7,1)(.8,.8)};
\draw[draw=black,fill=green!60!white, nearly transparent] plot [smooth cycle] coordinates {(.2,-.1)(.1,-.3)(.7,-1)(.8,-.8)};
\begin{scope}[scale = .7, shift={(5,0)}]

\node[circle,fill=black,inner sep=0pt,minimum size=4pt] at (0,0) {};
\node[circle,fill=black,inner sep=0pt,minimum size=4pt] at (-1,1) {};
\node[circle,fill=black,inner sep=0pt,minimum size=4pt] at (-2,0) {};
\node[circle,fill=black,inner sep=0pt,minimum size=4pt] at (-1,-1) {};
\node[circle,fill=black,inner sep=0pt,minimum size=4pt] at (1,-1) {};
\node[circle,fill=black,inner sep=0pt,minimum size=4pt] at (1,1) {};
\draw (0,0)--(-1,1);
\draw (0,0)--(1,1);
\draw (0,0)--(-1,-1);
\draw (0,0)--(1,-1);
\draw (-2,0)--(-1,1);
\draw (-2,0)--(-1,-1);

\draw (1.5,0) node{$=$};
\begin{scope}[shift={(4,0)}]
\node[circle,fill=black,inner sep=0pt,minimum size=4pt] at (0,0) {};
\node[circle,fill=black,inner sep=0pt,minimum size=4pt] at (-1,1) {};
\node[circle,fill=black,inner sep=0pt,minimum size=4pt] at (-2,0) {};
\node[circle,fill=black,inner sep=0pt,minimum size=4pt] at (-1,-1) {};
\node[circle,fill=black,inner sep=0pt,minimum size=4pt] at (1,-1) {};
\node[circle,fill=black,inner sep=0pt,minimum size=4pt] at (1,1) {};
\draw (0,0)--(-1,1);
\draw (0,0)--(1,1);
\draw (0,0)--(-1,-1);
\draw (0,0)--(1,-1);
\draw (-2,0)--(-1,1);
\draw (-2,0)--(-1,-1);
\end{scope}

\end{scope}
\end{tikzpicture}
\caption{The minimal homogeneous $\mathcal{I}_\mathcal{M}$-stratification of an open cover of a set of solutions to $y^2=x^3+x^2$. In this example, we obtain the trivial stratification.} 
\label{fig:elliptic}
\end{center}
\end{figure}
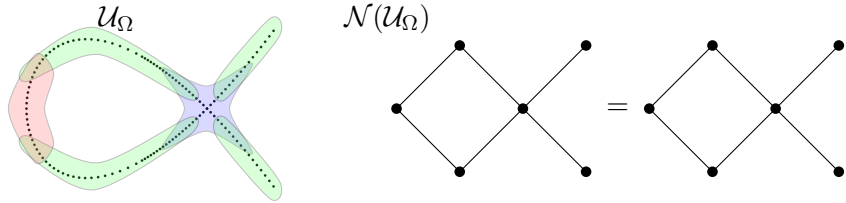

To contrast the $\mathcal{I}_\mathcal{M}$-stratification illustrated in Figure \ref{fig:elliptic}, we include (Figure \ref{fig:homology-elliptic}) a local homology stratification of the simplicial complex $\mathcal{N}(\mathcal{U}_\Omega)$. Unlike the polynomial based stratification, the local homology stratification distinguishes the singular point where the elliptic curve crosses over itself, as well as the boundary points of the simplicial complex.  

\begin{figure}[ht!]
\begin{center}
\begin{tikzpicture}[scale=1,circ/.style={
    circle,
    fill=white,
    draw,
    outer sep=0pt,
    inner sep=1.5pt
  }]

\draw (-1,1.5) node{$X$};
\node[circle,fill=black,inner sep=0pt,minimum size=4pt] at (0,0) {};
\node[circle,fill=black,inner sep=0pt,minimum size=4pt] at (-1,1) {};
\node[circle,fill=black,inner sep=0pt,minimum size=4pt] at (-2,0) {};
\node[circle,fill=black,inner sep=0pt,minimum size=4pt] at (-1,-1) {};
\node[circle,fill=black,inner sep=0pt,minimum size=4pt] at (1,-1) {};
\node[circle,fill=black,inner sep=0pt,minimum size=4pt] at (1,1) {};
\draw (0,0)--(-1,1);
\draw (0,0)--(1,1);
\draw (0,0)--(-1,-1);
\draw (0,0)--(1,-1);
\draw (-2,0)--(-1,1);
\draw (-2,0)--(-1,-1);

\draw (1.5,0) node{$=$};
\begin{scope}[shift={(4,0)}]
\draw (-.5,1.5) node{$X_1-X_0$};
\node[circle,fill=black,inner sep=0pt,minimum size=4pt] at (-1,1) {};
\node[circle,fill=black,inner sep=0pt,minimum size=4pt] at (-2,0) {};
\node[circle,fill=black,inner sep=0pt,minimum size=4pt] at (-1,-1) {};
\draw (0,0)--(-1,1);
\draw (0,0)--(1,1);
\draw (0,0)--(-1,-1);
\draw (0,0)--(1,-1);
\draw (-2,0)--(-1,1);
\draw (-2,0)--(-1,-1);
\node[circle,fill=black,inner sep=0pt,minimum size=4pt] at (0,0) {};
\node[circle,fill=white,inner sep=0pt,minimum size=3pt] at (0,0) {};
\node[circle,fill=black,inner sep=0pt,minimum size=4pt] at (1,-1) {};
\node[circle,fill=black,inner sep=0pt,minimum size=4pt] at (1,1) {};
\node[circle,fill=white,inner sep=0pt,minimum size=3pt] at (1,-1) {};
\node[circle,fill=white,inner sep=0pt,minimum size=3pt] at (1,1) {};
\draw (2,0) node{$\sqcup$};
\end{scope}
\begin{scope}[shift={(7,0)}]
\draw (0,1.5) node{$X_0$};
\node[circle,fill=black,inner sep=0pt,minimum size=4pt] at (0,0) {};
\node[circle,fill=black,inner sep=0pt,minimum size=4pt] at (1,-1) {};
\node[circle,fill=black,inner sep=0pt,minimum size=4pt] at (1,1) {};
\end{scope}
\end{tikzpicture}
\caption{The minimal homogeneous $\mathcal{L}$-stratification of a triangulation $X$ of a bounded connected set of real solutions to $y^2=x^3+x^2$.} 
\label{fig:homology-elliptic}
\end{center}
\end{figure}
\para{Points sampled from the sundial.}
In this example we study a higher dimensional corner, as a comparison to the second example in this section. We will assume that we have a point sample $\Omega$ of the sundial (see Figure \ref{fig:sundial}), with an open cover $\mathcal{U}$ of the sundial whose nerve is the triangulation given in Figure \ref{fig:triangulated-sundial-poly}. For example, for each vertex $\tau$ in the prescribed triangulation, define a cover element $U_\tau$ to be the image of the open star of $\tau$ under the triangulation homeomorphism of Section \ref{sec:background}. Assume that the point sample $\Omega$ and the set of polynomials $\mathcal{M}$ have the property that if $U$ is an intersection of open sets in the open cover, then $\mathcal{I}(U)=\mathcal{I}_\mathcal{M}(U\cap \Omega)$ (here $\mathcal{I}(U)$ denotes the set of all real valued polynomials which vanish on $U$). We will additionally assume that the base of the sundial is contained in a 2-dimensional plane in $\mathbb{R}^3$, with the complement of the base contained in a subspace perpendicular to the 2-dimensional plane. The resulting $\mathcal{I}_\mathcal{M}$-stratification is given in Figure \ref{fig:triangulated-sundial-poly}.
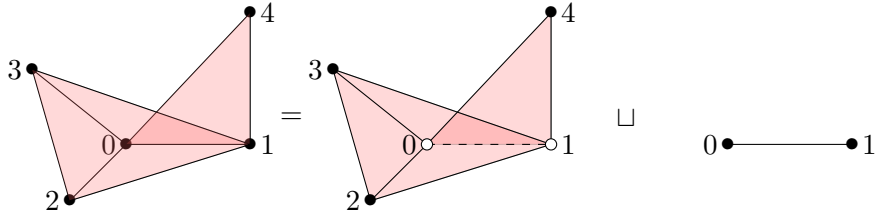
\begin{figure}[ht!]
\begin{center}
\begin{tikzpicture}[scale=1,circ/.style={
    circle,
    fill=white,
    draw,
    outer sep=0pt,
    inner sep=1.5pt
  }]
\draw (0,0) node {$\bullet$};
\draw (0,0) node[left] {0};
\draw (1.65,0) node {$\bullet$};
\draw (1.65,0) node[right] {1};
\draw (1.65,1.75) node {$\bullet$};
\draw (1.65,1.75) node[right] {4};
\draw (-1.25,1) node {$\bullet$};
\draw (-1.25,1) node[left] {3};
\draw (-.75,-.75) node {$\bullet$};
\draw (-.75,-.75) node[left] {2};
\draw (0,0)--(1.65,0);
\draw (0,0)--(1.65,1.75);
\draw (1.65,0)--(1.65,1.75);
\draw (1.65,0)--(-1.25,1);
\draw (1.65,0)--(-.75,-.75);
\draw (-1.25,1)--(-.75,-.75);
\draw (0,0)--(-.75,-.75);
\draw (0,0)--(-1.25,1);
\draw[draw=black,fill=red!60!white, nearly transparent] (0,0)--(1.65,0)--(1.65,1.75);
\draw[draw=black,fill=red!60!white, nearly transparent] (-.75,-.75)--(1.65,0)--(-1.25,1);

\draw (2.2,.35) node{$=$};
\draw[draw=black,fill=red!60!white, nearly transparent] (4+0,0)--(4+1.65,0)--(4+1.65,1.75);
\draw[draw=black,fill=red!60!white, nearly transparent] (4+-.75,-.75)--(4+1.65,0)--(4+-1.25,1);
\draw (4+0,0) node[left] {0};
\draw (4+1.65,0) node[right] {1};
\draw (4+1.65,1.75) node[right] {4};
\draw (4+-1.25,1) node[left] {3};

\draw (4+-.75,-.75) node[left] {2};
\draw[dashed] (4+0,0)--(4+1.65,0);
\draw  (4+0,0)--(4+1.65,1.75);
\draw (4+1.65,0)--(4+1.65,1.75);
\draw (4+1.65,0)--(4-1.25,1);
\draw  (4+1.65,0)--(4+-.75,-.75);
\draw  (4+-1.25,1)--(4+-.75,-.75);
\draw (4+0,0)--(4+-.75,-.75);
\draw (4+0,0)--(4+-1.25,1);
\draw (4+-.75,-.75) node{$\bullet$};
\draw (4+-1.25,1) node{$\bullet$};
\draw (4+1.65,1.75) node{$\bullet$};
\draw (4+1.65,0) node[circ]{};
\draw (4+0,0) node[circ]{};

\draw (6.65,.35) node{$\sqcup$};

\draw (8+0,0) node[left] {0};
\draw (8+1.65,0) node[right] {1};
\draw (8+0,0)--(8+1.65,0);

\draw (8+1.65,0) node{$\bullet$};
\draw (8+0,0) node{$\bullet$};

\end{tikzpicture}
\caption{A triangulated sundial and its minimal homogeneous $\mathcal{I}_\mathcal{M}$-stratification.} 
\label{fig:triangulated-sundial-poly}
\end{center}
\end{figure}

\para{Relation to the mapper construction.}
As a final example, we will show how $\mathcal{I}_\mathcal{M}$-stratifications naturally apply to the mapper construction. 
The mapper algorithm, originally developed in~\cite{SinghMemoliCarlsson2007}, gives a topological description of the fibers of a continuous function. We will illustrate the fundamental concept of the algorithm through an example. Suppose $\Omega$ is an \emph{$\epsilon$-net} of points on the torus $\mathbb{T}$ in $\mathbb{R}^3$ (illustrated in Figure \ref{fig:mapper}). In other words, let $\Omega$ be any finite subset of $\mathbb{T}$ such that the Hausdorff distance between $\Omega$ and $\mathbb{T}$ is less than $\epsilon$. Let
 $$\Omega_\epsilon=\{x\in\mathbb{R}^3:\text{min}_{\omega\in\Omega}||x-\omega||<\epsilon\}$$ be the $\epsilon$-thickening of $\Omega$, and suppose $f$ is a continuous map from $\Omega_\epsilon$ to $\mathbb{R}$. Let $\mathcal{U}$ be an open cover of $f(\Omega_\epsilon)$. The mapper construction of $\mathcal{U}$ and $f$, denoted $\mathcal{M}(\mathcal{U},f)=\mathcal{N}(f^\ast(\mathcal{U}))$, is the nerve of the connected pull back $f^\ast(\mathcal{U})$ of the open cover $\mathcal{U}$, where 
$$f^\ast(\mathcal{U})=\{U\subset\Omega_\epsilon:\text{$U$ is a connected component of }f^{-1}(V)\text{ for some }V\in\mathcal{U}\}.$$

\begin{figure}[ht!]
\begin{center}
\includegraphics[width=0.90\linewidth]{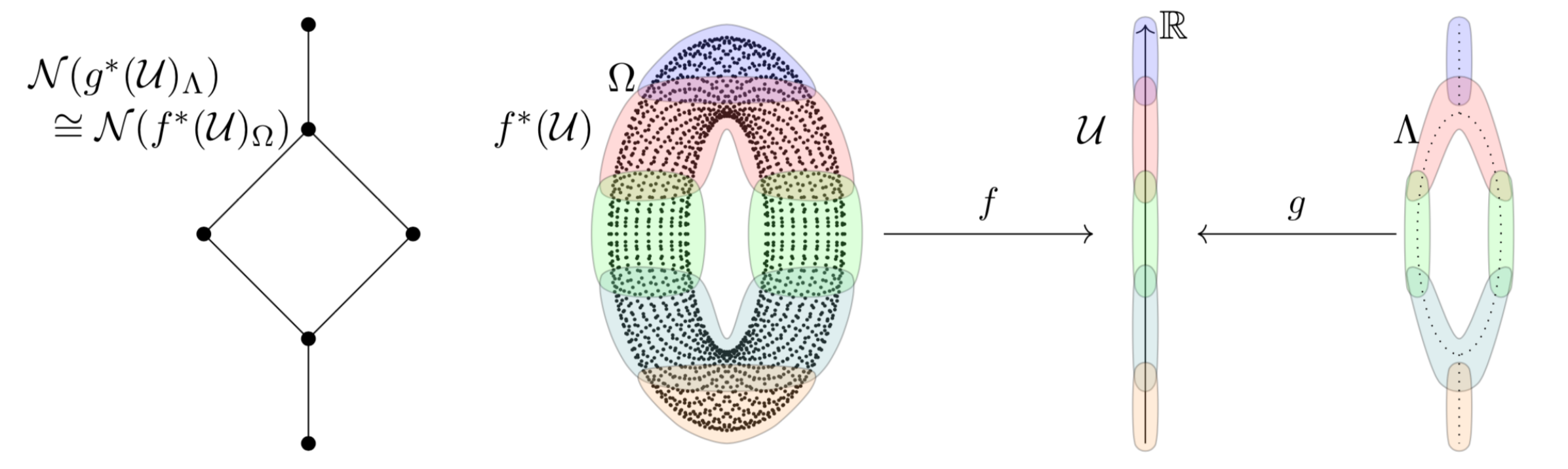} 
\caption{The mapper construction applied to the height function on the torus as well as the height function on an ellipse with two branching lines.} 
\label{fig:mapper}
\end{center}
\end{figure}

The $\mathcal{I}_\mathcal{M}$-stratification is naturally suited to work with the mapper construction studied in~\cite{SinghMemoliCarlsson2007,MunchWang16,CarriereOudot17}. Let $f^\ast(\mathcal{U})_\Omega = \{U\cap \Omega: U\in f^\ast(\mathcal{U})\}$ be the open cover of the point set $\Omega$ (in the discrete topology) induced by $f^\ast(\mathcal{U})$. Additionally, choose $\mathcal{M}$ so that $R_\mathcal{M}$ is the set of all polynomials in three variables with degree less than or equal to 4,
$$
R_\mathcal{M} = \{ax^{n_1}+by^{n_2}+cz^{n_3}:a,b,c\in\mathbb{R}\text{ and }0\le n_1,n_2,n_3\le 4\}.
$$
As an example, we will contrast the $\mathcal{I}_\mathcal{M}$-stratification obtained from $\Omega$, $f$, and $\mathcal{U}$ with the $\mathcal{I}_\mathcal{M}$-stratification obtained from $\Lambda$, $g$, and $\mathcal{U}$, where $\Lambda$ is an $\epsilon$-net of points sampled from the ellipse with two branching lines embedded in $\mathbb{R}^3$ (illustrated in Figure \ref{fig:mapper}), and $g$ is the height function on $\Lambda_\epsilon$. To distinguish the pre-sheaf of vanishing polynomials defined using $\Omega$ and $f$ from the pre-sheaf of vanishing polynomials defined using $\Lambda$ and $g$, we will use the notation $\mathcal{I}_\mathcal{M}^\Omega$ and $\mathcal{I}_\mathcal{M}^\Lambda$, respectively. We can consider the $\mathcal{I}_\mathcal{M}^\Omega$-stratification of $\mathcal{N}(f^\ast(\mathcal{U})_\Omega)$, by taking each open set $V\in f^\ast(\mathcal{U})_\Omega$ to the vector space $I_\mathcal{M}(V)$ of polynomials in $\mathcal{M}$ which vanish on $V$. The $\mathcal{I}^\Omega_\mathcal{M}$-stratification of $\mathcal{N}(f^\ast(\mathcal{U})_\Omega)$ will differ from the $\mathcal{I}^\Lambda_\mathcal{M}$-stratification of $\mathcal{N}(g^\ast(\mathcal{U})_\Lambda)$, even though the underlying topological spaces are homeomorphic. Suppose $R$ and $r$ are the radii of the torus from which the points in $\Omega$ are sampled. For each open set $V\in f^\ast(\mathcal{U})_\Omega$, $\mathcal{I}^\Omega_\mathcal{M}(V)=\mathbb{R}\langle (x^2+y^2+z^2+R^2-r^2)^2-4R^2(x^2+y^2)\rangle$, resulting in the trivial stratification. 
\begin{figure}[ht!]
\begin{center}
\begin{tikzpicture}[scale=.8,circ/.style={
    circle,
    fill=white,
    draw,
    outer sep=0pt,
    inner sep=1.5pt
  }]
  
\begin{scope}[scale = 1, shift={(-4,0)}]
\draw (-1.5,1) node {$\mathcal{N}(f^\ast(\mathcal{U})_\Omega)$};
\draw (1,0)--(0,1);
\draw (-1,0)--(0,1);
\draw (0,1)--(0,2);
\draw (1,0)--(0,-1);
\draw (-1,0)--(0,-1);
\draw (0,-1)--(0,-2);
\node[circle,fill=black,inner sep=0pt,minimum size=4pt] at (-1,0) {};
\node[circle,fill=black,inner sep=0pt,minimum size=4pt] at (1,0) {};
\node[circle,fill=black,inner sep=0pt,minimum size=4pt] at (0,1) {};
\node[circle,fill=black,inner sep=0pt,minimum size=4pt] at (0,2) {};
\node[circle,fill=black,inner sep=0pt,minimum size=4pt] at (0,-1) {};
\node[circle,fill=black,inner sep=0pt,minimum size=4pt] at (0,-2) {};
\end{scope}
\node at (-2,0) {$=$};

\draw (1,0)--(0,1);
\draw (-1,0)--(0,1);
\draw (0,1)--(0,2);
\draw (1,0)--(0,-1);
\draw (-1,0)--(0,-1);
\draw (0,-1)--(0,-2);
\node[circle,fill=black,inner sep=0pt,minimum size=4pt] at (-1,0) {};
\node[circle,fill=black,inner sep=0pt,minimum size=4pt] at (1,0) {};
\node[circle,fill=black,inner sep=0pt,minimum size=4pt] at (0,1) {};
\node[circle,fill=black,inner sep=0pt,minimum size=4pt] at (0,2) {};
\node[circle,fill=black,inner sep=0pt,minimum size=4pt] at (0,-1) {};
\node[circle,fill=black,inner sep=0pt,minimum size=4pt] at (0,-2) {};
\end{tikzpicture}
\caption{The minimal homogeneous $\mathcal{I}^\Omega_\mathcal{M}$-stratification of $\mathcal{N}(f^\ast(\mathcal{U})_\Omega)$.} 
\label{fig:torus-poly}
\end{center}
\end{figure}
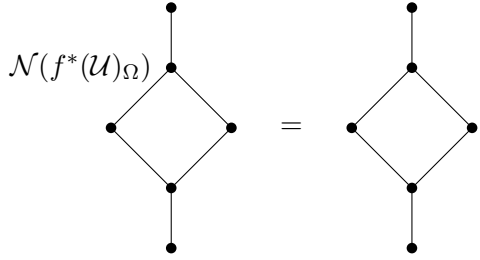
 However, open sets $V\in g^\ast(\mathcal{U})_\Lambda$ containing only points on the branching lines will result in vector spaces generated (as a vector space) by polynomials (in $R_\mathcal{M}$) of the form $ x\cdot p(x,y,z)$. Alternatively, open sets $V\in g^\ast(\mathcal{U})_\Lambda$ which contain only points on the ellipse will result in the vector space $\mathbb{R}\langle cx^2+dy^2-1\rangle$ for constants $c,d\in\mathbb{R}$. Therefore, the $\mathcal{I}^\Lambda_\mathcal{M}$-stratification (depicted in Figure~\ref{fig:reeb-poly}) is nontrivial.
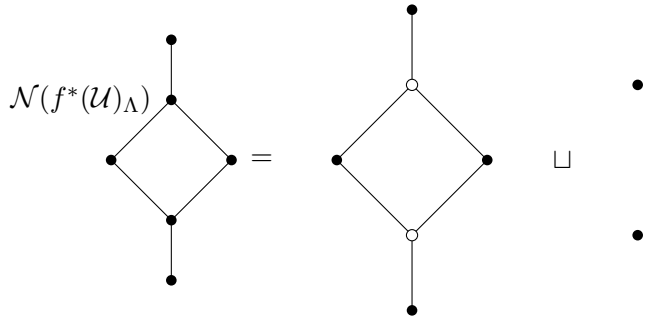
\begin{figure}[ht!]
\begin{center}
\begin{tikzpicture}[scale=1,circ/.style={
    circle,
    fill=white,
    draw,
    outer sep=0pt,
    inner sep=1.5pt
  }]
  
\begin{scope}[ shift={(-4,0)}]
\draw (-1.5,1) node {$\mathcal{N}(f^\ast(\mathcal{U})_\Lambda)$};
\node[circle,fill=black,inner sep=0pt,minimum size=4pt] at (-1,0) {};
\node[circle,fill=black,inner sep=0pt,minimum size=4pt] at (1,0) {};
\node[circle,fill=black,inner sep=0pt,minimum size=4pt] at (0,1) {};
\node[circle,fill=black,inner sep=0pt,minimum size=4pt] at (0,2) {};
\node[circle,fill=black,inner sep=0pt,minimum size=4pt] at (0,-1) {};
\node[circle,fill=black,inner sep=0pt,minimum size=4pt] at (0,-2) {};
\draw (1,0)--(0,1);
\draw (-1,0)--(0,1);
\draw (0,1)--(0,2);
\draw (1,0)--(0,-1);
\draw (-1,0)--(0,-1);
\draw (0,-1)--(0,-2);
\end{scope}
\node at (-2,0) {$=$};

\draw (1,0)--(0,1);
\draw (-1,0)--(0,1);
\draw (0,1)--(0,2);
\draw (1,0)--(0,-1);
\draw (-1,0)--(0,-1);
\draw (0,-1)--(0,-2);
\node[circle,fill=black,inner sep=0pt,minimum size=4pt] at (-1,0) {};
\node[circle,fill=black,inner sep=0pt,minimum size=4pt] at (1,0) {};
\node[circ] at (0,1) {};
\node[circle,fill=black,inner sep=0pt,minimum size=4pt] at (0,2) {};
\node[circ] at (0,-1) {};
\node[circle,fill=black,inner sep=0pt,minimum size=4pt] at (0,-2) {};
\node at (2,0) {$\sqcup$};
\node[circle,fill=black,inner sep=0pt,minimum size=4pt] at (3,1) {};
\node[circle,fill=black,inner sep=0pt,minimum size=4pt] at (3,-1) {};
\end{tikzpicture}
\caption{The minimal homogeneous $\mathcal{I}^\Lambda_\mathcal{M}$-stratification of $\mathcal{N}(g^\ast(\mathcal{U})_\Lambda)$.} 
\label{fig:reeb-poly}
\end{center}
\end{figure}

\section{Proofs of Our Main Results}
\label{sec:proofs}

We detail the proofs of our main theorems, that is, the existence of $\Fcal$-stratifications (Proposition~\ref{theorem:exists}), the existence of coarsest $\Fcal$-stratifications (Theorem~\ref{theorem:coarsest}), and the existence and uniqueness of minimal homogeneous $\Fcal$-stratifications (Theorem~\ref{theorem:uniqueness}). 

\subsection{Proof of Proposition~\ref{theorem:exists}}
\begin{proof}
We can take the finest filtration of $X$, so that each $X_i-X_{i-1}$ consists of a single point (i.e.~element) $S_i=X_i-X_{i-1}=x_i\in X$. Then $X=\coprod_{x_i\in X} x_i$. In order to insure that the corresponding filtration is a filtration by closed subsets, we need to order our points so that $x_i$ is minimal (with the poset ordering) in the complement of $X_{i-1}$. Now we wish to show that $\Fcal\vert_{x_i}$ is locally constant for each $x_i\in X$. This is trivial, and in fact $\Fcal\vert_{x_i}$ is a constant sheaf, since it is a sheaf defined on a topological space consisting of a single point. Therefore $\Fcal$ is constructible with respect to the single point decomposition $X=\coprod_{x_i\in X} x_i$.  
 
 \end{proof}

\subsection{Proof of Theorem~\ref{theorem:coarsest}}
 \label{ap:coarsest}
 \begin{proof}
This theorem can be proved immediately by noticing that there are only finitely many stratifications of $X$ ($X$ being a finite $T_0$-space with finite many points). Since the set of $\Fcal$-stratifications is nonempty, there must be an $\Fcal$-stratification with a minimal number of strata pieces, and such a stratification must be a coarsest $\Fcal$-stratification. However, for the purposes of developing an algorithm, we will prove this constructively by defining each $X_i$ in a coarsest $\Fcal$-stratification. Let $d_0$ be the dimension of $X$ and define $X_{d_0}:=X$. Define
$$
S_{d_0}:=\{x\in X_{d_0} : \Fcal(B_w\subset B_y)\text{ is an isomorphism for all chains }x\le y\le w\}
$$
Set $d_1$ to be the dimension of $X_{d_0}-S_{d_0}$. Then define $X_{d_1}$ to be the complement of $S_{d_0}$ in $X_{d_0}$:
$$X_{d_1}:=X_{d_0}-S_{d_0}
$$
Now each $d_0+1$ chain in $X_{d_0}$ terminates with an element $x$ of $S_{d_0}$ because $\Fcal\vert_{B_x}$ is automatically constant when $x$ is the terminal element of a maximal chain. The dimension of $X_{d_1}$ is strictly less than $d_0$, since each $d_0+1$ chain in $X$ ends with an element of $S_{d_0}$, and thus is not a chain in $X-S_{d_0}$. Define $X_i:=X$ for each $i$ such that $d_1<i< d_0$. Now $X_{d_1}$ is itself a finite $T_0$-space. Let $B^{d_1}_x$ denote the minimal open neighborhood of $x$ in $X_{d_1}$. Then we can use the same condition as above to define $S_{d_1}$:
$$
S_{d_1}:=\{x\in X_{d_1} : \Fcal(B_w\subset B_y)\text{ is an iso. for all chains }x\le y \le w\text{ in }X_{d_1}\}
$$
Again notice that $S_{d_1}$ is not empty since terminal elements of maximal chains are guaranteed to be elements of $S_{d_1}$. Continue to define $d_i$ to be the dimension of $X_{d_{i-1}}-S_{d_{i-1}}$ and $X_{d_i}:=X_{d_{i-1}}-S_{d_{i-1}}$ inductively until $d_i=0$. To fill out the missed indices, define $S_j$ to be empty if $d_i< j < d_{i-1}$ and $X_j:=X_{d_i}$ if $d_i\le j< d_{i-1}$. 

Notice that each $S_i$ is an open subset of $X_i$. Therefore $X_{i-1}$ is closed in $X_i$, and $S_i$ is an open set in $X_i$ (and therefore locally closed in $X$). So we have constructed a stratification of $X$.

Now we will focus on showing that $\Fcal\vert_{S_i}$ is locally constant. If $S_i$ is non-empty, then $S_i=S_{d_k}$ for some $k$. If we want to show that $\Fcal\vert_{S_i}$ is locally constant, we need to check that $(\Fcal\vert_{S_i})\vert_{B^i_x}$ is locally constant for each $x\in S_i$ (where $B^i_x=B_x\cap X_i$). Consider the presheaf $\mathcal{E}$ on $B^i_x$, which maps each open set $U\subset B^i_x$ to $\Fcal(B_x)$, and each morphism $U\subset V$ to the identity morphism. So we have $\mathcal{E}(U)=\Fcal(B_x)$ for all $U\subset B^i_x$, and $\mathcal{E}(U\subset V)=\text{id}:\Fcal(B_x)\rightarrow \Fcal(B_x)$ for all $U\subset V\subset B_x^i$. Notice that the sheafification of $\mathcal{E}$ is by definition a constant sheaf. Let $\mathcal{E}'$ be the presheaf on $B^i_x$ defined by $\mathcal{E}'(U)=\Fcal(\str(U))$ and $\mathcal{E}'(U\subset V)=\Fcal(\str(U)\subset \str(V))$. Notice that the sheafification of $\mathcal{E}'$ is by definition $(\Fcal\vert_{S_i})\vert_{B^i_x}$. We  want to show that the sheafification of $\mathcal{E}$ is isomorphic to the sheafification of $\mathcal{E}'$. Recall that it is enough to show that $\mathcal{E}$ and $\mathcal{E}'$ agree on minimal open sets $B_y^i$, and give the same restriction maps between minimal open sets. We have the equalities (as morphisms) $\mathcal{E}'(B^i_y\subset B^i_w)=\Fcal(B_y\subset B_w)=\Fcal(B_x\subset B_x)=\mathcal{E}(B^i_y\subset B^i_w)$, which we obtain by applying our definition of $\mathcal{E}'$, the assumption (made in our definition of $S_i$) that $\Fcal(B_y\subset B_w)$ is an isomorphism for all $x\le y\le w\in X_i$, and the definition of $\mathcal{E}$. These equalities further imply that $\mathcal{E}'(B^i_y)=\mathcal{E}(B_y^i)$. So we have shown that the sheafification of $\mathcal{E}$ is isomorphic to the sheafification of $\mathcal{E}'$, which is a constant sheaf. Therefore $(\Fcal\vert_{S_i})\vert_{B_x}$ is constant, which implies that $\Fcal\vert_{S_i}$ is locally constant, which implies that $\Fcal$ is constructible with respect to the decomposition $X=\coprod S_i$. So we have constructed an $\Fcal$-stratification. 
 
 Now suppose that there exists a coarser $\Fcal$-stratification
 $$
 \emptyset\subset X'_0\subset\cdots\subset X'_n=X
 $$
 We will continue by using the notation $S^\circ_i$ (respectively $S'_j{}^\circ $) to denote a connected component of $S_i$ (respectively $S_j'$). Suppose $S^\circ_i\subsetneq S'_j{}^\circ$. Let $x\in S'_j{}^\circ-S^\circ_i$. Since $\Fcal$ is locally constant when restricted to $S'_j{}^\circ$, we have that $\Fcal$ is constant when restricted to $B_x\cap S'_j{}^\circ$. Notice that $B_x\cap S^\circ_i\subset B_x\cap S'_j{}^\circ$. Therefore $\Fcal$ is constant when restricted to $B_x\cap S^\circ_i$. Since $S^\circ_i$ is an open subset of $X_i$, we have that $B_x\cap X_i=B_x\cap S^\circ_i$. So we can finally conclude that $\Fcal$ is constant when restricted to $B_x\cap X_i$. However, by the definition of $S_{i}$ above, we see that $x$ must be an element of $S_i$. Therefore $S_i^\circ\subset S'_j{}^\circ\subset S_i$, which implies that $S_i^\circ =S'_j{}^\circ$. Therefore each stratum piece of the stratification $\emptyset\subset X_0\subset\cdots\subset X_n=X$ is equal to a stratum piece of the stratification $\emptyset\subset X'_0\subset\cdots\subset X'_n=X$. So we can conclude that these two stratifications are equivalent, which concludes the proof.  
\end{proof}

\subsection{Proof of Theorem~\ref{theorem:uniqueness}}
\label{ap:uniqueness}
 \begin{proof}
We will prove this constructively by defining each $X_i$ in a minimal homogeneous $\Fcal$-stratification, and then showing that any minimal homogeneous $\Fcal$-stratification is necessarily equal to the stratification constructed below. In many ways, this proof is similar to the proof of Theorem \ref{theorem:coarsest}. Let $d_0$ be the dimension of $K$ and $X_{d_0}=X$. Define
$$
H_{d_0}:=\{x\in X_{d_0} : \text{Cl}(B_x)\text{ is homogeneous of dimension $d_0$}\}
$$ ($H$ for homogeneous) and 
$$
C_{d_0}:=\{x\in H_{d_0} : \Fcal(B_w\subset B_y)\text{ is an iso. for all }x\le y\le w\}
$$ 
($C$ for constant) where $\text{Cl}(B_x)=\{y\in X_{d_0}:y\le s\text{ for some }s\in B_x\}$ is the closure of $B_x$. Then define $S_{d_0}=H_{d_0}\cap C_{d_0}$. Set ${d_1}$ to be the dimension of $X_{d_0}-S_{d_0}$. Then define $X_{d_1}$ to be $X_{d_0}-S_{d_0}$. Now each ${d_0}+1$ chain in $X_{d_0}$ terminates with an element $x$ of $S_{d_0}$ because Cl$(x)$ is homogeneous of dimension ${d_0}$ by our assumption that $X_{d_0}$ consists of simplices of a simplicial complex. We have that ${d_1}$ is strictly less than ${d_0}$, since each ${d_0}+1$ chain in $X_{d_0}$ ends with an element of $S_{d_0}$, and thus is not a chain in $X_{d_1}$. Define $X_i:=X_{d_1}$ for each $i$ such that ${d_1}<i< {d_0}$. Now $X_{d_1}$ is itself a finite $T_0$-space. Let $B^{d_1}_x$ denote the minimal open neighborhood of $x$ in $X_{d_1}$. Then we can use the same condition as above to define $$
H_{d_1}:=\{x\in X_{d_1} : \text{Cl}(B^{d_1}_x)\text{ is homogeneous of dimension ${d_1}$}\}
$$
and
$$
C_{d_1}:=\{x\in H_{d_1} : \Fcal(B_w\subset B_y)\text{ is an iso. for all chains }x\le y\le w\text{ in }H_{d_1}\}
$$ 
As before, let $S_{d_1}=H_{d_1}\cap C_{d_1}$, and notice that $S_{d_1}$ is not empty since $X_{d_1}$ corresponds to a sub-simplicial complex of $K$. Continue to define $H_{d_k}$, $C_{d_k}$, $S_{d_k}$, and $X_{d_{k+1}}$ inductively until $d_k=0$. To fill out the missed indices, define $S_i$ to be empty if $d_j< i < d_{j-1}$ and $X_i:=X_{d_j}$ if $d_j<i<d_{j-1}$. 
 
 Notice that each $S_i$ is an open subset of $X_i$. Therefore $X_{i-1}$ is closed in $X_i$, and $X_i-X_{i-1}$ is an open set in $X_i$ (and therefore locally closed in $X$). Additionally, $\text{Cl}(X_i-X_{i-1})=\text{Cl}(S_i)$ is either empty or is homogeneous of dimension $i$. So we have constructed a homogeneous stratification of $X$. Now we wish to show that this is a homogeneous $\Fcal$-stratification. It remains to show that $\Fcal\vert_{S_i}$ is locally constant. This follows the same argument as in the proof of Theorem \ref{theorem:coarsest}. So $\Fcal$ is constructible with respect to the stratification given by the filtration $0\subset X_0\subset\cdots\subset X_{d_0}=X$. We will denote this stratification by $\mathfrak{X}$. \par
 Suppose that there exists a minimal homogeneous $\Fcal$-stratification
 $$
 \emptyset\subset X'_0\subset\cdots\subset X'_{d_0}=X
 $$
 denoted by $\mathfrak{X}'$. Now $S'_i$ must contain all of the elements of $X'_i$ which correspond to $i$-simplices in $K$. Moreover, for each element $x\in S'_i$, there exists $y\in X'_i$ corresponding to an $i$-simplex in $K$, such that $x\le y$ (due to the homgeneity of $X'_i-X'_{i-1}$). Suppose $a\in S_n$ and $b\in S'_n$ such that $a\le b$ (an analogous argument follows for $b\le a$). Since $b$ is necessarily a face of an $n$-simplex $\tau\in X'_{n}=X_n=X$, we have $a\le b \le \tau$. Since $\tau$ is a an $n$-simplex, we have that $\tau\in S_n$. Since $\Fcal$ is assumed to be locally constant when restricted to $S_i$ and $S'_i$, we have that $\rho_{ B_a, B_\tau}$ and $\rho_{B_b, B_\tau}$ are isomorphisms. By the third part of the sheaf axiom (which is also true for presheaves), we have that $\rho_{B_b,B_\tau }\circ\rho_{ B_a,B_b}=\rho_{B_a,B_\tau} $. Therefore, $\Fcal(B_b\subset B_a)$ is an isomorphism. So if we set $S''_n:=S_n\cup S'_n$, then $\text{Cl}(S''_n)$ is homogeneous of dimension $n$ and $\Fcal\vert_{S''_n}$ is locally constant. However, by our construction of $S_n$, we can see that $S_n$ is the maximal set with these properties. So $S_n\subset S''_n$ implies that $S_n=S''_n$. This implies that $S'_n\subset S_n=X-X_{n-1}$. If $S'_n\subsetneq S_n$, then we would have that $\mathfrak{X}<\mathfrak{X}'$, which would contradict the minimality of $\mathfrak{X}'$. So we must have that $S_n=S_n'$, which implies that $X_{n-1}'=X_{n-1}$. This allows us to inductive use the same argument above to show that $X_i'=X_i$ for all $i$. Therefore the two stratifications are equal, which concludes the proof.  
\end{proof}

\section{Discussion}
\label{sec:discussion}

Many problems in computational geometry and topology are solved by finding suitable combinatorial models which reflect the geometric or topological properties of a particular space of interest. One way to study properties of a space such as the pinched torus in Figure \ref{fig:pinchedtorus} is to begin by finding a triangulation of the space, which in some sense provides a combinatorial model which is amenable to computation. The corresponding triangulation can be thought of as a stratification of the space, by defining the $d$-dimensional stratum to be the collection of $d$-dimensional simplices. However, this stratification is too ``fine'' in a sense, as it breaks up the underlying space into too many pieces, resulting in each stratum piece retaining relatively little information about the underlying geometry of the total space. The results of this paper can be interpreted as a method for computing a coarsening of the stratification obtained from the triangulation of our underlying space, using homological techniques (or more generally, sheaf-theoretic techniques) to determine when two simplices should belong to the same coarser stratum. 

There are two key features of the sheaf-theoretic stratification learning algorithm which should be highlighted. The first feature is that we avoid computations which require the sheafification process. At first glance this may be surprising to those not familiar with cellular sheaves, since constructible sheaves can not be defined without referencing sheafification, and our algorithm builds a stratification for which a given sheaf is constructible. 
In other words, each time we want to determine the restriction of a sheaf to a subspace, we need to compute the sheafification of the presheaf referenced in the definition of the pull back of a sheaf (Section \ref{subsec:sheaves}). 
We can avoid this by noticing two facts. Suppose $\mathcal{E}$ is a presheaf and $\mathcal{E}^+$ is the sheafification of $\mathcal{E}$. First, in the setting of finite $T_0$-spaces, we can deduce if $\mathcal{E}^+$ is constant by considering how it behaves on minimal open neighborhoods. Second, the behavior of $\mathcal{E}^+$ will agree with the behavior of the presheaf $\mathcal{E}$ on minimal open neighborhoods. Symbolically, this is represented by the equalities 
$\mathcal{E}^+(B_x)= \mathcal{E}(B_x)$ and $ \mathcal{E}^+(B_w\subset B_x)=\mathcal{E}(B_w\subset B_x)$ for all pairs of minimal open neighborhoods $B_w\subset B_x$ (where $B_x$ is a minimal open neighborhood of $x$, and $B_w$ is a minimal open neighborhood of an element $w\in B_x$). Therefore, we can determine if $\mathcal{E}^+$ is constant, locally constant, or constructible, while only using computations involving the presheaf $\mathcal{E}$ applied to minimal open neighborhoods. 

The second feature of our algorithm (which is made possible by the first) is that the only sheaf-theoretic computation required is checking if $\Fcal(B_w\subset B_x)$ is an isomorphism for each pair $B_w\subset B_x$ in our space. This is extremely relevant for implementations of the algorithm, as it minimizes the number of expensive computations required to build an  $\Fcal$-stratification. For example, if our sheaf is the local homology sheaf, we will only need to compute the restriction maps between local homology groups of minimal open neighborhoods. The computation of local homology groups can therefore be distributed and computed independently. Additionally, once we have determined whether the local homology restriction maps are isomorphisms, we can quickly compute a coarsest $\Lcal$-stratification, or a minimal homogeneous $\Lcal$-stratification, without requiring any local homology groups to be recomputed. As we saw with the pre-sheaf of vanishing polynomials $\mathcal{I}_\mathcal{M}$, there are often cases of stratifications which rely only on the image of the sheaf applied to each minimal open set, making computations even more streamlined. 

There are several interesting questions related to $\Fcal$-stratifications that we will investigate in the future. 
We are interested in the stability of local homology based stratifications under refinements of triangulations of polyhedra. In this direction, it would be useful to view $\Lcal$-stratifications from the perspective of persistent homology. If we are given a point cloud sampled from a compact polyhedron, it would be natural to ask about the asymptotics of persistent local homology based stratifications. 

\section*{Acknowledgements}
This work was partially supported by NSF IIS-1513616 and NSF ABI-1661375. 

\bibliographystyle{alpha}
\bibliography{sheaf}


\end{document}